\documentclass[lettersize,journal,10pt]{IEEEtran}
\usepackage{amsmath,amsfonts}
\usepackage{algorithmic}
\usepackage{algorithm}
\usepackage{array}
\usepackage[caption=false,font=normalsize,labelfont=sf,textfont=sf]{subfig}
\usepackage{textcomp}
\usepackage{stfloats}
\usepackage{url}
\usepackage{verbatim}
\usepackage{graphicx}
\usepackage{cite}
\usepackage[inline]{enumitem}

\usepackage{amsmath,amsthm}
\usepackage{amssymb}
\usepackage{xcolor}
\usepackage{graphicx}
\usepackage[colorlinks=true, allcolors=blue]{hyperref}
\usepackage{pgfplots}
\pgfplotsset{compat=newest}

\usepackage{tikz}
\usetikzlibrary{arrows.meta,backgrounds,calc,shadings,shapes.arrows,shapes.symbols,shadows,fit,positioning,topaths,hobby}
\usepackage[normalem]{ulem}
\useunder{\uline}{\ul}{}
\usepackage{booktabs}
\usepackage{balance}

\usepackage[capitalize]{cleveref}

\newtheorem{proposition}{Proposition}%
\newtheorem{theorem}{Theorem}%
\newtheorem{lemma}{Lemma}%

\newtheorem{definition}{Definition}%
\newtheorem{example}{Example}%
\newtheorem{approach}{Approach}%
\newtheorem{remark}{Remark}

\definecolor{bleudefrance}{rgb}{0.19, 0.55, 0.91}

\newcommand{\cD}{\mathcal{D}}

\newcommand{\rvX}{\mathsf{X}}
\newcommand{\rvY}{\mathsf{Y}}
\newcommand{\rvZ}{\mathsf{Z}}

\newcommand{\F}{\mathbb{F}}

\newcommand{\entropy}[2]{\ensuremath{\mathrm{H}_{#1}\left( #2 \right)}}

\newcommand{\mutinf}[2]{\ensuremath{\mathrm{I}_{#1}\left( #2 \right)}}

\definecolor{my_blue}{rgb}{0,0,0.8}
\definecolor{my_darkblue}{rgb}{0.2,0.2,0.6}
\definecolor{my_lightblue}{rgb}{0,0,0.8}
\definecolor{shadecolor}{gray}{.65} 
\definecolor{my_red}{rgb}{0.63922,0.14902,0.21961}
\definecolor{my_redgray}{rgb}{0.5,0.14902,0.21961}
\definecolor{my_lightred}{rgb}{0.63922,0.14902,0.21961}
\definecolor{my_green}{rgb}{0.2,0.6,0.25}
\definecolor{my_lightgreen}{rgb}{0.5,0.9,0.55}
\definecolor{blond}{rgb}{0.98, 0.94, 0.75}
\definecolor{caucasian}{rgb}{1,0.8,0.6}
\definecolor{richelectricblue}{rgb}{0.03, 0.57, 0.82}
\definecolor{bleudefrance}{rgb}{0.19, 0.55, 0.91}
\definecolor{ogreen}{rgb}{0,0.5,0}
\definecolor{xgray}{rgb}{0.9, 0.9, 0.9}

\definecolor{TUMBlue}{RGB}{0,101,189} %
\definecolor{TUMBlueDark}{RGB}{0,82,147} %
\definecolor{TUMBlueLight}{RGB}{152,198,234} %
\definecolor{TUMBlueMedium}{RGB}{100,160,200} %

\definecolor{TUMIvory}{RGB}{218,215,203} %
\definecolor{TUMGreen}{RGB}{162,173,0} %
\definecolor{TUMGray}{gray}{0.6} %

\definecolor{TUMGreenDark}{RGB}{0,124,48} %
\definecolor{TUMRed}{RGB}{196,7,27} %

\newcommand{\nclients}{\ensuremath{n}}
\newcommand{\client}{\ensuremath{i}}
\newcommand{\clientidxb}{\ensuremath{i_1}}
\newcommand{\clientidxc}{\ensuremath{i_2}}
\newcommand{\objective}{\ensuremath{t}}
\newcommand{\queriedobjective}{\ensuremath{j}}
\newcommand{\nobjectives}{\ensuremath{T}}
\newcommand{\colluder}{\ensuremath{i}}
\newcommand{\ncolluders}{\ensuremath{z}}
\newcommand{\ncolluderss}{\ensuremath{z_s}}
\newcommand{\ncolludersq}{\ensuremath{z_q}}
\newcommand{\nlabels}{\ensuremath{s}}
\newcommand{\labelidx}{\ensuremath{\ell}}
\newcommand{\logit}[1][\labelidx]{\ensuremath{y_{\client, #1}^{(\objective)}}}
\newcommand{\genlogit}[1][\labelidx]{\ensuremath{y_{ #1}^{(\objective)}}}
\newcommand{\logitqueried}[1][\labelidx]{\ensuremath{y_{\client, #1}^{(\queriedobjective)}}}
\newcommand{\partidx}{\ensuremath{u}}

\newcommand{\logitpart}[1][\partidx]{\ensuremath{{y}_{\client, \partition,#1}^{(\objective)}}}
\newcommand{\logitpartentry}[1][\partidx]{\ensuremath{{y}_{\client, \partition,#1}^{(\objective)}}}
\newcommand{\logitpartrv}[1][\partidx]{\ensuremath{\rvY_{\client, \partition, \partidx}^{(\objective)}}}
\newcommand{\logitpartentryqueried}[1][\partidx]{\ensuremath{{y}_{\client, \partition,#1}^{(\queriedobjective)}}}
\newcommand{\logitpartqueried}[1][\partidx]{\ensuremath{y_{\client, \partition, #1}^{(\queriedobjective)}}}
\newcommand{\logitpartqueriedl}[1][\labelidx]{\ensuremath{y_{\client, #1}^{(\queriedobjective)}}}
\newcommand{\logitpartqueriedrv}[1][\partidx]{\ensuremath{\mathrm{Y}_{\client, \partition, #1}^{(\queriedobjective)}}}
\newcommand{\rlogits}[1][\colluder]{\ensuremath{r_{\client, \partition, #1}^{(\objective)}}}
\newcommand{\rlogitsl}{\ensuremath{r_{\client, \labelidx}^{(\objective)}}}

\newcommand{\rlogitsqueriedl}[1][\colluder]{\ensuremath{r_{\client, \labelidx}^{(\queriedobjective)}}}

\newcommand{\nclasses}{\ensuremath{c}}
\newcommand{\npartitions}{\ensuremath{P}}
\newcommand{\partition}{\ensuremath{p}}
\newcommand{\ssclientl}[1][x]{\ensuremath{f_{\client, \labelidx}^{(\objective)}(#1)}}
\newcommand{\ssclientbl}[1][x]{\ensuremath{f_{\clientidxb, \labelidx}^{(\objective)}(#1)}}
\newcommand{\ssclientcl}[1][x]{\ensuremath{f_{\clientidxc, \labelidx}^{(\objective)}(#1)}}
\newcommand{\ssclient}[1][x]{\ensuremath{f_{\client, \partition}^{(\objective)}(#1)}}

\newcommand{\sumssclient}[1][x]{\ensuremath{F_{\partition}^{(\objective)}(#1)}}
\newcommand{\sumssclientl}[1][x]{\ensuremath{F_{\labelidx}^{(\objective)}(#1)}}

\newcommand{\storagecode}{\ensuremath{\mathcal{C}}}
\newcommand{\querycode}{\ensuremath{\mathcal{D}}}
\newcommand{\kstorage}{\ensuremath{k_\storagecode}}

\newcommand{\kquerydual}{\ensuremath{k_{\querycode^\perp}}}
\newcommand{\grs}[1]{\ensuremath{\text{GRS}_{#1}}}
\newcommand{\productcode}{\ensuremath{\storagecode \star \querycode}}

\newcommand{\ssquery}[1][x]{\ensuremath{q_{\partition}^{(\objective, \queriedobjective)}(#1)}}
\newcommand{\ssqueryl}[1][x]{\ensuremath{q_{\labelidx}^{(\objective, \queriedobjective)}(#1)}}
\newcommand{\querymessage}[1][x]{\ensuremath{\delta_{\partition}^{(\objective, \queriedobjective)}}}
\newcommand{\querymessagel}[1][x]{\ensuremath{\delta_{\labelidx}^{(\objective, \queriedobjective)}}}
\newcommand{\rquery}[1][\colluder]{\ensuremath{k_{\partition, #1}^{(\objective)}}}
\newcommand{\rqueryl}{\ensuremath{k_{\labelidx}^{(\objective)}}}
\newcommand{\ranswer}[1][\colluder]{\ensuremath{a_{\partition, #1, \client}^{(\objective)}}}

\newcommand{\partitionset}{\ensuremath{\mathcal{P}}}

\newcommand{\function}[1][\objective]{\ensuremath{h_{#1}}}

\newcommand{\data}[1][\client]{\ensuremath{\mathcal{D}_{#1}}}
\newcommand{\model}[2]{\ensuremath{w_{#1, #2}}}
\newcommand{\publicdata}{\ensuremath{\mathcal{D}_\text{pub}}}
\newcommand{\publicsample}{\ensuremath{\mathbf{x}_\labelidx}}

\newcommand{\sumo}{\ensuremath{\sum_{\objective=1}^\nobjectives}}
\newcommand{\sumc}{\ensuremath{\sum_{\client=1}^\nclients}}

\newcommand{\answer}[1][\client]{\ensuremath{\mathcal{A}^{#1}}}
\newcommand{\answerl}[1][\client]{\ensuremath{\mathcal{A}^{#1}_\labelidx}}
\newcommand{\ssanswerpart}[1][\partition]{\ensuremath{A_{#1}^\client(x)}}
\newcommand{\answerpart}[1][\partition]{\ensuremath{A_{#1}^\client(\alpha_\client)}}

\newcommand{\graph}{\ensuremath{\mathcal{G}}}
\newcommand{\vertices}{\ensuremath{\mathcal{V}}}
\newcommand{\edges}{\ensuremath{\mathcal{E}}}
\newcommand{\incident}{\ensuremath{\mathbf{I}}}
\newcommand{\edge}[1][\objective]{\ensuremath{e_{#1}}}
\newcommand{\incidentedge}[1][\objective]{\ensuremath{\mathcal{I}(\edge[#1])}}
\newcommand{\rep}{\ensuremath{\rho}}
\newcommand{\incidentvertex}[1][\client]{\mathcal{I}(#1)}
\newcommand{\rlogitssum}[1][\colluder]{\ensuremath{r_{\partition, #1}^{\prime (\objective)}}}
\newcommand{\rlogitssumqueried}[1][\colluder]{\ensuremath{r_{\partition, #1}^{\prime (\queriedobjective)}}}
\newcommand{\define}[1][\colluder]{\ensuremath{\triangleq}}

\newcommand{\colss}{\ensuremath{\mathcal{T}_s}}
\newcommand{\colquery}{\ensuremath{\mathcal{T}_q}}

\newcommand{\objectivetmp}{\ensuremath{\objective^\prime}}
\renewcommand{\entropy}[1]{\ensuremath{\mathrm{H}\left(#1\right)}}
\newcommand{\condentropy}[2]{\ensuremath{\mathrm{H}\left(#1 \mid #2\right)}}
\renewcommand{\mutinf}[2]{\ensuremath{\mathrm{I}\left(#1; #2\right)}}
\newcommand{\condmutinf}[3]{\ensuremath{\mathrm{I}\left(#1; #2 \mid #3\right)}}
\newcommand{\queryrv}[1][\client]{\mathrm{Q}_{\partition, #1}^{(\objective)}}
\newcommand{\queriesrv}[1][\mathcal{T}]{\mathrm{Q}_{#1}^{(\objective)}}
\newcommand{\queryrvcol}[1][\colquery]{\mathrm{Q}_{\partition, #1}^{(\objective)}}
\newcommand{\queryrvcoltmp}[1][\client]{\mathrm{Q}_{\partition, \colquery}^{(\objective^\prime)}}
\newcommand{\logitrv}[1][\labelidx]{\ensuremath{\rvY_{\client, #1}^{(\objective)}}}
\newcommand{\alllogitsrv}[1][\labelidx]{\ensuremath{\rvY_{\client, [\nlabels]}^{(\objective)}}}
\newcommand{\logitqueriedrv}[1][\labelidx]{\ensuremath{\rvY_{\client, #1}^{(\queriedobjective)}}}

\newcommand{\allreceived}[1][\clientidxb]{\ensuremath{\mathcal{M}_{#1}}}
\newcommand{\allreceivedcol}[1][\ncolluderss]{\ensuremath{\mathcal{M}_{#1}}}

\newcommand{\allclient}{\ensuremath{\mathcal{S}_\client}}
\newcommand{\allcol}[1][\colquery]{\ensuremath{\mathcal{S}_{#1}}}

\newcommand{\ssrandpart}[1][\partition]{\ensuremath{R_{#1}(x)}}
\newcommand{\rrand}[1][\colluder]{\ensuremath{s_{\partition, #1}}}
\newcommand{\randpart}[1][\partition]{\ensuremath{R_{#1}(\alpha_\client)}}

\newcommand{\ssanswerpartsym}[1][\partition]{\ensuremath{A_{#1}^\prime(x)}}
\newcommand{\answerpartsym}[1][\partition]{\ensuremath{A_{#1}^\prime(\alpha_\client)}}

\newcommand{\sumi}{\ensuremath{\sum_{\objective \in \incidentvertex}}}
\newcommand{\dualco}{\ensuremath{\nu_{\objective, \client}}}
\newcommand{\dualcoqueried}[1][\client]{\ensuremath{\nu_{\queriedobjective, #1}}}
\newcommand{\evalc}[1][\client]{\ensuremath{\alpha_{#1}}} 
\newcommand{\eval}{\ensuremath{\alpha}}

\newcommand{\evalidx}{\ensuremath{\vartheta}}
\newcommand{\sumio}{\ensuremath{\sum_{\client \in \incidentedge}}}
\newcommand{\sumioqueried}{\ensuremath{\sum_{\client \in \incidentedge[\queriedobjective]}}}

\newcommand{\logitpartqueriedsum}[1][\partidx]{\ensuremath{\bar{y}_{\partition, #1}}}
\newcommand{\evalsum}[1][\evalidx]{A^{(#1)}}

\newcommand{\camready}[1]{\textcolor{black}{#1}}
\newcommand{\rev}[1]{\textcolor{black}{#1}}

\begin{document}

\title{Federated One-Shot Learning with Data Privacy and Objective-Hiding}

\author{
\IEEEauthorblockN{Maximilian Egger\IEEEauthorrefmark{1},~\IEEEmembership{Student Member},~IEEE, Rüdiger Urbanke\IEEEauthorrefmark{2},~\IEEEmembership{Senior Member},~IEEE, and~Rawad~Bitar\IEEEauthorrefmark{1},~\IEEEmembership{Member},~IEEE\\}
\IEEEauthorblockA{\IEEEauthorrefmark{1}%
                   Technical University of Munich, Germany \{maximilian.egger, rawad.bitar\}@tum.de}
                   
\IEEEauthorblockA{\IEEEauthorrefmark{2}%
                   École Polytechnique Fédérale de Lausanne, Switzerland \{rudiger.urbanke\}@epfl.ch}
\thanks{This project is funded by DFG (German Research Foundation) projects under Grant Agreement Nos. BI 2492/1-1 and WA 3907/7-1.}
\thanks{Part of the work was done when RB and ME visited RU at EPFL supported in parts by EuroTech Visiting Researcher Programme grants.}%
\thanks{A preliminary version of this paper is accepted for presentation at the IEEE International Symposium on Information Theory 2025 \cite{egger2025federatedisit}.}}

\markboth{%
}%
{Shell \MakeLowercase{\textit{et al.}}: A Sample Article Using IEEEtran.cls for IEEE Journals}

\maketitle

\begin{abstract}
Privacy in federated learning is crucial, encompassing two key aspects: safeguarding the privacy of clients' data and maintaining the privacy of the federator's objective from the clients. While the first aspect has been extensively studied, the second has received much less attention.

We present a novel approach that addresses both concerns simultaneously, drawing inspiration from techniques in knowledge distillation and private information retrieval to provide strong information-theoretic privacy guarantees.

Traditional private function computation methods could be used here; however, they are typically limited to linear or polynomial functions. To overcome these constraints, our approach unfolds in three stages. In \textbf{stage 0}, clients perform the necessary computations locally. In \textbf{stage 1}, these results are shared among the clients, and in \textbf{stage 2}, the federator retrieves its desired objective without compromising the privacy of the clients’ data. The crux of the method is a carefully designed protocol that combines secret-sharing-based multi-party computation and a graph-based private information retrieval scheme. %
We show that our method outperforms existing tools from the literature when properly adapted to this setting.

\end{abstract}


\section{Introduction}

We consider federated learning (FL), a framework where a federator and a set of clients with private data collaborate to train a neural network. Due to privacy constraints, the clients' data cannot be directly shared with the federator or among the clients. This privacy concern has been extensively studied in the literature~\cite{bonawitz2017practical,kairouz2021advances,li2021survey,wen2023survey,ye2023heterogeneous}.  
There exists a second, often overlooked, privacy concern: ensuring the privacy of the federator's objective used to train the neural network. This aspect has not been explored in the literature to the same extent.\footnote{The notion of intention-hiding has only appeared recently in a different setting in vertical FL \cite{tang2023ihvfl}, where the intention of model training is implemented using suitable data preprocessing in the form of a private set intersection.} 

We present a novel approach that ensures the privacy of the clients' data and simultaneously hides the objective of the federator through a careful combination of a secure aggregation method and a tailored private information retrieval (PIR) scheme. 
The key challenge of the overall problem arises from the complexity of the computations required for training the neural network and the inherent heterogeneity of the clients' data. E.g., training a neural network in a distributed manner typically requires each client to compute a gradient of a loss function with respect to the current neural network using their private data. This is a highly non-linear computation.

We pose a very general research question: \textit{How can a federator use the clients' private data to accomplish a task, while hiding their objective and maintaining the privacy of the clients' data?} 

An instantiation of this problem is fine-tuning a large-language model for one target objective out of many target objectives known to the clients. We do not impose any assumptions on the clients' data; hence, the computed function might be different for each client. %
If the task were linear, the standard technique of {\em secure aggregation} \cite{bonawitz2017practical} could be employed effectively. However, for non-linear tasks, such as training a neural network, averaging the final models reached by the clients fails to produce a meaningful model. Even when the clients' data is similar, their resulting models may differ significantly, and the averaged model may lack meaningful utility. %

The challenge of combining the knowledge of multiple non-linear models is commonly referred to as (federated) knowledge distillation, particularly in the context of multiple teacher models and a single student model \cite{liu2020adaptive}. In this framework, pre-trained teacher models—representing the clients' models, or function outputs—transfer their knowledge to the student model, which corresponds to the federator's model. We draw inspiration from these concepts to enable arbitrary function computations within model training processes. Additionally, the recently introduced concept of auditing for private prediction \cite{chadha2024auditing} highlights the importance of exploring privacy-preserving techniques in this domain.

Our work can be viewed as a generalization of private function computation, a well-established framework for outsourcing complex computations while preserving the privacy of the function being evaluated (e.g., \cite{karpuk2018private,raviv2019private,sun2018capacity}). This computation can be performed on datasets stored either centrally or distributively \cite{tandon2018pir}. A stronger notion than function privacy is the additional protection of the underlying data used for computation, as discussed in \cite{raviv2019private}. This enhanced privacy is achieved using techniques such as secret sharing. While privacy guarantees can be categorized into information-theoretic and computational approaches, this work focuses exclusively on the former, \rev{ensuring privacy even against adversaries with unlimited computational power}.

The foundational concept underpinning private function computation is private information retrieval (PIR) \cite{mirmohseni2018private,obead2018achievable}. In PIR, a dataset is distributed across one or more servers, and a client retrieves a specific file or subset of data without revealing which file is of interest. Numerous studies have explored PIR from different angles, including single-server PIR, PIR with replicated data \cite{chor1998private}, PIR with MDS-coded data \cite{obead2018achievable}, PIR with secretly shared data \cite{freij2017private}, and PIR using graph-based replication, also referred to as non-replicated data storage \cite{raviv2019graph,jia2020asymptotic,banawan2019private}. 
This concept has been further extended to symmetric privacy, which ensures that the client learns nothing about the dataset beyond the specific file requested from the databases \cite{wang2019symmetric}. For a comprehensive overview of recent advancements and open challenges in PIR, we refer readers to the surveys \cite{ulukus2022private,d2024guided}.

Beyond private function computation \cite{raviv2019private,zhu2022symmetric,obead2022private-jsac,karpuk2018private,tahmasebi2020private,sun2018capacity,tahmasebi2019private}, concepts from PIR have been extended to private function retrieval \cite{mirmohseni2018private,obead2018achievable,obead2022private,zhang2020bounds,jia2023x,jia2024asymptotic,esmati2021multi,heidarzadeh2020private}, private inner product retrieval \cite{mousavi2019private}, and private linear transformation \cite{kazemi2021multi,heidarzadeh2022single}. While these frameworks address private computation for linear functions, polynomial functions, or compositions of linear functions, they do not provide solutions for computing arbitrary functions. This limitation introduces new challenges, which we address in this work. Our approach can be seen as a generalized framework for the computation of arbitrary predictors, extending beyond the previously studied settings.

We propose an end-to-end solution for federated one-shot learning that ensures both data and objective privacy. Under the assumption of a limited number of colluding clients, our approach prevents the leakage of private data to other clients or the federator. Specifically targeting federated learning applications, we introduce a novel method that integrates concepts from graph-based (symmetric) private information retrieval, secret sharing, multi-party computation, coded storage, knowledge distillation, and ensemble learning.

Our solution requires a public unlabeled dataset accessible to all clients and the federator, a common and non-restrictive assumption in semi-supervised machine learning problems \cite{learning2006semi}. In Stage 0, each client is assigned a subset of objectives chosen from a pool of candidate objectives. For each assigned objective, the client trains a local model and uses it to label the shared public dataset. In Stage 1, clients use a carefully tailored secret sharing scheme to share the labels privately among each other and aggregate the received shares. Since secret sharing schemes are linear, the aggregate of the secret shares consists of shares of the aggregated labels. In Stage 2, the federator uses a symmetric PIR scheme to receive only the aggregated labels corresponding to their objective of interest. Thus, enabling the reconstruction of the federator's model while leveraging data contributions from all clients and maintaining client privacy by observing only aggregated labels. Additionally, the federator's objective of interest remains hidden from the clients through the use of a PIR scheme. %

We focus on one-shot federated learning for several key reasons: (1) iterative schemes could compromise the privacy of the objective, (2) iterative methods incur significant communication overhead due to privacy mechanisms, and (3) collaborative iterative training using a shared public dataset introduces additional challenges. Further details are provided in \cref{remark:one_shot_fl}.

Our contributions are summarized as follows:

\begin{itemize} \item We formulate the general problem and propose a three-stage solution comprising the task assignment stage, the sharing stage and the query stage, with jointly designed codes to minimize the overall communication cost.
\item Building on \cite{jia2020asymptotic}, we leverage tools from the duals of Reed-Solomon codes to design a flexible task assignment scheme. Then, we develop a graph-based PIR scheme tailored for the designed private coded storage used in the query stage. This approach generalizes the storage pattern to ramp Secret Sharing (e.g., McEliece-Sarwate Secret Sharing \cite{mceliece1981sharing}), enhancing efficiency in generating and storing shares of the clients' labels. 
\item We extend the framework to incorporate data privacy against the federator, i.e., ensuring no additional information beyond the desired function is leaked to the federator.
\item We evaluate the rate of our scheme, demonstrating significant gains over existing PIR methods for graph-based coded data when computational resources are constrained. Additionally, we propose an optimized scheme utilizing star-product PIR for scenarios where computation is inexpensive. \end{itemize}

\begin{remark} In scenarios where client privacy is not a primary concern, a simple approach is to independently download all labels from each client and aggregate the results at the federator. For cryptographic guarantees, symmetric privacy—protecting the clients' results beyond the desired objective—can be achieved through oblivious transfer protocols for individual queries. Our proposed solution adheres to a stronger notion of information-theoretic privacy, safeguarding both individual client data, by only revealing the aggregate of the labels of the objective of interest, as well as the privacy of the objective itself. \end{remark}

\section{Related Work}
We review the following fundamental ingredients upon which our scheme is based on: graph-based PIR, secure aggregation in FL, private function computation, symmetric PIR and knowledge distillation.

\paragraph{PIR} While there has been an abundance of works, we specifically mention those closest to our interest: PIR for MDS coded data was studied, e.g., in \cite{freij2017private}. Joint message encoding for PIR was studied in \cite{sun2019breaking}, and~\cite{tian2020storage} studied the trade-off storage and download cost in PIR. Although the latter two are conceptually different, the ideas are loosely related. We focus in the following on graph-based PIR, whose methodology is most related to parts of our contribution.

\paragraph{Graph-Based PIR}
The problem of PIR on graph-based data storage was first introduced in \cite{raviv2019graph}, in which the replication of files is modeled by (hyper-)graphs, where vertices correspond to storage nodes and (hyper-)edges correspond to files and connect nodes storing the same file. The proposed scheme is proved to be uncritical in regard to collusions as long as the graph exhibits non-cyclic structures. However, privacy of the stored data is not considered. A similar concept has arisen concurrently and termed PIR for non-replicated databases \cite{banawan2019private}, later extended to optimal message sizes \cite{keramaati2020private}. Many follow-up works have considered PIR on different graph structures, e.g., \cite{sadeh2021bounds} where bounds on the capacity of PIR were derived for specific graphs, in particular for the star-graph (with one universal node storing all the messages), and the complete graph. A linear programming-based bound was given for general graphs. Follow-up works studies the capacity for a $K=4$ star graph \cite{yao2023capacity}. Privacy of the data through secret sharing was studied in \cite{jia2019cross} by means of Shamir Secret Sharing for $X$-secure and $T$-private PIR in a non-graph setting. Cross-space interference alignment is used to improve the rate of the PIR scheme by efficiently returning multiple information symbols of interest per query. This has later been extended to graph-based PIR for messages encoded by a Shamir Secret Sharing \cite{jia2020asymptotic}, which can be seen as the extension of \cite{raviv2019private} to $X$-security. The dual code of a Generalized Reed Solomon is used together with cross-subspace alignment for inference cancellation. Private function computation in the graph-based setting has recently been studied in \cite{jia2024asymptotic} for $X$-security and $T$-privacy. A generalization of cross-subspace alignment using algebraic geometry codes \cite{makkonen2024algebraic} with secret sharing \cite{makkonen2024secret} was recently proposed. Those works are concerned with Shamir Secret Sharing, and hence not suitable for this problem. Graph-based secret sharing was independently studied in \cite{de2023bounds}.

Semantic PIR in which the length of the messages may be different was considered in \cite{vithana2021semantic}. In \cite{karpuk2018private}, the authors studied PIR for replicated data and general functions such that the query space is a vector-space. Theorem~2 in \cite{sun2018capacity} and a result in \cite{obead2019capacity} appear to consider a general set of functions from non-colluding replicated databases. 
In \cite{tandon2018pir}, the authors consider a setting where the user decides how to store the messages on the non-colluding databases. There is no privacy of the data from the servers.

\paragraph{Secure Aggregation in FL}
Secure aggregation for FL was first introduced in~\cite{bonawitz2017practical}. Follow-up works are concerned with the communication overhead of such methods, e.g., \cite{bell2020secure,so2021turbo,kadhe2020fastsecagg,so2022lightsecagg,jahani2022swiftagg+}. Alternative models are considered in, e.g., \cite{schlegel2023coded,sami2023secure,egger2024private}.

\paragraph{Private Function Computation and Distributed Computing}
Private function computation was extensively studied in the literature for linear and polynomial functions, cf. \cite{raviv2019private,zhu2022symmetric,obead2022private-jsac,karpuk2018private,tahmasebi2020private,sun2018capacity,tahmasebi2019private}. Tools from PIR were further applied to distributed computing. In \cite{ulukus2022private}, a review and survey on this topic is provided. For instance, polynomial computation from distributed data was studied in \cite{tan2023privacy}, and distributed matrix multiplication from MDS coded storage was studied in \cite{zhu2023information}. Tools from PIR have further been used in FL for private submodel learning termed private read update write \cite{vithana2023private}.

\paragraph{Symmetric PIR}
The capacity of symmetric PIR (SPIR) where the messages and the randomness are encoded with codes with different parameters are considered in \cite{wang2019mismatched}. Symmetric PIR from MDS coded data with potentially colluding databases was considered in \cite{wang2019symmetric}. SPIR with user-side common randomness was considered in \cite{wang2021symmetric,wang2022digital}. Random SPIR was introduced in \cite{wang2022digital}, where the user is interested in a random message rather than a specific one. Closer to our work, symmetric private polynomial computation was studied in \cite{zhu2022symmetric}, where the authors also consider a finite set of candidate polynomial functions. Related to symmetric privacy with multiple servers, the combinations of multiple oblivious transfer protocols was recently studied in \cite{farras2023one}.

\paragraph{Knowledge Distillation in FL}
Throughout this work, we make use of ensemble learning methods, which have been extensively studied, e.g., in \cite{kittler1998combining,kuncheva2014combining,polikar2006ensemble}, and for heterogeneous classifiers in \cite{vongkulbhisal2019unifying}. The concept is also related to the student-teacher model in the setting of multiple teachers \cite{liu2020adaptive}. We refer interested readers to the survey in \cite{li2024federated} for an extensive review of knowledge distillation in FL.

\section{Preliminaries and System Model}

\textbf{Notation.} We denote finite fields of cardinality $q$ by $\mathbb{F}_q$. For a natural number $a$ we define the following set notation $[a] \define \{1, \cdots a\}$. For a random variable $\rvX$, we refer to the entropy as $\entropy{\rvX}$, and for two random variables $\rvX$ and $\rvY$, we denote the mutual information as $\mutinf{\rvX}{\rvY}$. Similarly, we denote the conditional entropy and the conditional mutual information conditioned on a random variable $\rvZ$ as $\condentropy{\rvX}{\rvZ}$ and $\condmutinf{\rvX}{\rvY}{\rvZ}$, respectively.

\textbf{Private Function Computation.} 
Each client $\client \in [\nclients]$ holds a private dataset $\mathcal{D}_\client$. In some cases, these datasets consist of distinct samples drawn from the same underlying distribution, referred to as the homogeneous case. In other cases, the datasets are drawn from different distributions, referred to as the heterogeneous case.

Let $\mathcal{F} = \{\function\}_{t=1}^{T}$ denote the pool of $\nobjectives$ candidate functions (or objectives)\rev{, public, and hence, known to all parties}. Among these, the federator is interested in a specific function $\function[\queriedobjective]$, where the index $j$ is unknown to the clients. Specifically, the federator aims to compute $\function[\queriedobjective]\left(\bigcup_{\client \in [\nclients]} \mathcal{D}_\client\right),$
leveraging the combined data from all clients, $\bigcup_{\client \in [\nclients]} \mathcal{D}_\client$. 
Our solution is designed for scenarios where $\function[\queriedobjective]$ is additively separable, i.e., 
$
\function[\queriedobjective]\left(\bigcup_{\client \in [\nclients]} \mathcal{D}_\client\right) = \sum_{\client \in [\nclients]} \function[\queriedobjective](\mathcal{D}_\client).
$

This property is naturally satisfied by linear models, and as we will demonstrate in the next section, it can also be extended to certain non-linear models, including neural networks.

\textbf{Non-linear Training Processes.} 
We now expand on the previous paragraph and discuss how standard tools from knowledge distillation can ensure additive separability even in cases where highly non-linear training processes are involved. The core idea is that the clients train models to label a public unlabeled dataset for each objective. The function $\function$ is then the composition of the local model training and the local model applied to the public data.

More precisely, let $\model{\client}{\objective}$ be the local model trained at client $\client$ for objective $\objective$. Let $\publicdata$ denote the public dataset consisting of $\nlabels$ samples $\{\publicsample\}_{\labelidx \in [\nlabels]}$. Define 
$
\function[\objective](\data) = \model{\client}{\objective}(\publicdata),
$
where 
$
\model{\client}{\objective}(\publicdata) \define \{\model{\client}{\objective}(\publicsample)\}_{\labelidx \in [\nlabels]} = \{\logit\}_{\labelidx \in [\nlabels]}.
$
Here, $\logit, \labelidx \in [\nlabels]$, refers to the label of client $\client$ and objective $\objective$ for data sample $\labelidx$ in $\publicdata$. Note that $\logit$ is of dimension $\nclasses$, and each entry is quantized to at most $\lfloor q / \nclients + 1 \rfloor$ levels. In summary, in this case, $\function$ is a function that takes as input a training dataset, trains a model, and outputs the labels for a fixed public dataset, i.e.,
$\function: \data[] \mapsto \{\genlogit\}_{\labelidx \in [\nlabels]}.
$

\textbf{Task Assignment.} %
Although the best performance can be reached if all clients compute a function (or train a model) for each objective, training a model for all $\nobjectives$ objectives might be expensive. Hence, for each $\objective$, we assign the task of computing $\function$ to only a subset of the clients. We model this assignment of tasks by a hypergraph $\graph(\vertices, \edges)$ with clients $[\nclients]$ represented by vertices, and objectives represented by hyperedges $\edges$. The binary incident matrix $\incident \in [0,1]^{\nclients \times \nobjectives}$ has its $(\client,\objective)$ entry equal to $1$ if client $\client$ computes the model for objective $\objective$, and $0$ otherwise. We assume a symmetric setup, where the column weight of $\incident$ is constant and equal to $\rep$, i.e., exactly $\rep$ clients compute a model for each objective $\objective$. We denote by $\incidentedge$ the set of clients incident with hyperedge $\edge$. Hence, all clients $\client \in \incidentedge$ compute a model for the objective $\objective$. Further, let $\incidentvertex$ denote the set of all edges incident with client $\client$.

\textbf{Privacy Guarantees.} During the execution of the protocol, clients share messages amongst each other and with the federator. Let $\allreceived[\client]$ be the set of all messages received by client $\client$, and for a set $\mathcal{T}\subseteq [\nclients]$, let $\allreceivedcol[\mathcal{T}]$ be the set of all messages received by all clients in $\mathcal{T}$, i.e., $\allreceivedcol[\mathcal{T}] \define \{\allreceived[\client]\}_{\client \in \mathcal{T}}$. Let $\queriesrv[\mathcal{T}]$ be the set of all queries received from the federator by clients $\client \in \mathcal{T}$ for objective $\objective$. After the sharing stage is complete, let $\allclient$ be the data stored by client $\client$, and $\allcol[\mathcal{T}]$ be the data stored by all clients $\client \in \mathcal{T}$ and let $\logitrv$ be the random variable representing the prediction of the public sample $\publicsample$ using objective $\objective$. %
We consider information-theoretic privacy notions, defined formally in the sequel, assuming at most $\ncolluderss$ colluding clients that target compromising other clients' individual data privacy and at most $\ncolludersq$ clients trying to infer the federator's objective. %
\rev{The multifold privacy guarantees for the clients' data and the federator's objective are formally stated in \cref{def:privacy_clients,def:objective_hiding,def:privacy_federator}}.
\begin{definition}[Data Privacy from Clients] \label{def:privacy_clients}
    No client's data is leaked to any other set of at most $\ncolluderss$ colluding clients, i.e., for all $\client \in [\nclients]$ and any client collusion set $\colss \subset [\nclients]\setminus\{\client\}, \vert \colss \vert \leq \ncolluderss$,
    \begin{equation*}
        \condmutinf{\{\function(\data)\}_{\objective \in \incidentvertex, \camready{\client \in [\nclients]}}}{\allreceivedcol[\colss]}{\{\function(\data)\}_{\objective \in \incidentvertex, \client \in \colss}} = 0.
    \end{equation*}
    where $\{\function(\data)\}_{\objective \in \incidentvertex} = \{\logitrv\}_{\objective \in \incidentvertex, \labelidx \in [\nlabels]}$ for the special case of one-shot FL.
\end{definition}

Since the clients' training data might be correlated, the conditioning ensures that nothing further is leaked beyond what is known to the colluding clients from their own computations.

\begin{definition}[Objective-Hiding] \label{def:objective_hiding}
    The identity $\queriedobjective$ of the function of interest to the federator is private from any $\ncolludersq$ colluding clients, i.e., for each $\colquery \subset [\nclients], \vert \colquery \vert \leq \ncolludersq$, it holds that
    \begin{equation*}
        \mutinf{\{\queriesrv[\colss]\}_{\objective\in [\nobjectives], }, \allcol\camready{, \allreceivedcol[\colquery]}}{\queriedobjective} = 0.
    \end{equation*}
\end{definition}

Going beyond the above privacy measures, we further require that no information beyond the aggregate clients' function results is revealed to the federator. With $\answer$ being the answer received by the federator from client $\client$, and $\answer[{[\nclients]}]$ the set of answers from all clients $\client \in [\nclients]$, we have the following definition of data privacy against the federator.
\begin{definition}[Data Privacy from Federator] \label{def:privacy_federator}
    The federator's knowledge about the clients data is limited to the quantity $\sumc \function[\queriedobjective](\data)$ of interest, i.e.,
    \begin{align*}
        \condmutinf{\answer[{[\nclients]}], \{\queriesrv[{[\nclients]}]\}_{\objective\in [\nobjectives]}}{\{\function(\data)\}_{\objective \in \incidentvertex, \client \in [\nclients]}}{\!\!\! \sumioqueried \!\!\! \function[\queriedobjective](\data)} = 0,
    \end{align*}
    where for FL we have $\function(\data) = \{\logitrv\}_{\labelidx \in [\nlabels]}$, and $\sumioqueried \function[\queriedobjective](\data) = \{\sumioqueried \logitqueriedrv\}_{\labelidx \in [\nlabels]}$.
\end{definition}
In PIR, \cref{def:objective_hiding} corresponds to the user privacy, i.e., hiding the identity of the queried file, and \cref{def:privacy_federator} is referred to as the symmetric privacy guarantee. The communication cost is formally determined as the size of all transmitted messages, i.e., $\entropy{\allreceivedcol[{[\nclients]}], \answer[{[\nclients]}]}$. The task assignment stage incurs no communication overhead. Therefore, we will analyze ${R}_{\mathrm{share}}$ and ${R}_{\mathrm{PIR}}$, the rates of the sharing and query (PIR) stages of our scheme, which are formally defined next.
    \begin{align*} 
        {R}_{\mathrm{share}} & \!=\! \frac{\entropy{\sumioqueried \! \function[\queriedobjective](\data)}}{\camready{\sumc \entropy{\allreceived[{\client}]}}\!}, {R}_{\mathrm{PIR}} \!=\! \frac{\entropy{\sumioqueried \! \function[\queriedobjective](\data)}}{\camready{\sumc \entropy{\answer[{\client}]}}\!}
    \end{align*}

\section{Problem Illustration Through the Lens of Fine-tuning Large-Language Models}

To illustrate the problem and the principle idea of our solution, we take the example of fine-tuning large language models (LLMs), which recently gained significant attention through the progress and capabilities of generative models such as GPT-4, Llama 3 and Mistral 7B. Imagine a pre-trained and generic LLM suitable for a variety of tasks. The federator is interested in fine-tuning the model according to a specific objective $\queriedobjective$, for instance, sentiment analysis \cite{socher2013recursive}, article classification \cite{zhang2015character} or question classification \cite{li2002learning}. Note that both functions (models) and datasets can differ across objectives. While the clients have suitable labeled data at hand that can be used for supervised learning, their data should be kept private. Knowing the different objectives (or functions), each client $\client$ can individually fine-tune the LLM with respect to every possible objective $\objective \in [\nobjectives]$, including the objective $\queriedobjective$ of interest. Thereby, each client obtains a model $\model{\client}{\objective}$ for each objective of interest. Since the average of clients' models for the same objective trained on their individual data is not meaningful, we make use of a public and unlabelled dataset $\publicdata$, consisting of $\nlabels$ samples $\publicsample,\ \labelidx \in [\nlabels]$, to transfer the knowledge from clients to the federator. Each client $\client$ labels the public samples $\publicsample$ for each objective $\objective$, i.e., each trained LLM $\model{\client}{\objective}$, thereby obtaining $\nlabels$ labels $\logit$ for each $\objective \in [\nobjectives]$.

Since the individual labels contain sensitive information about the clients' data \cite{tang2023reducing}, their privacy is as important as that of the fine-tuned LLM models. For the federator, it is beneficial to receive the average of the predictions in order to reconstruct a suitable model~\cite{hinton2015distilling}. In fact, such averaged predictions for all samples of the public dataset might improve the performance due to the diversity of the clients' data, leading to better generalization results \cite{hinton2015distilling}. To make the federator obtain the sum of all predictions without revealing any information about individual labels to the federator or to other clients, we borrow ideas from multi-party computation~\cite{cramer2015secure}. Here, each client stores a message, and they collaboratively want to compute the sum of the messages without leaking individual ones to any subset of $\ncolluderss$ clients trying to compromise privacy. In FL, such concepts are well-known as the secure aggregation of local models (or gradients) at each iteration~\cite{bonawitz2017practical,jahani2022swiftagg+}. Such secure aggregation techniques lead to high computational overheads, especially when using basic versions of secret sharing such as Shamir Secret Sharing \cite{shamir1979how}.

\begin{example} \label{ex:description} To illustrate the challenges, we explain and contrast two approaches to solve this motivating example. Assume for simplicity unbounded compute power of the clients, i.e., each client computes the function result for all objectives, and let $\rep=\nclients = 5$ and $\ncolluderss = \ncolludersq = \ncolluders = 1$. The two approaches are: \begin{enumerate*}[label={\textit{(\roman*)}}]
    \item designing a scheme for our framework using known techniques from knowledge distillation, secure aggregation, and methods from the PIR literature~\cite{jia2020asymptotic}; and \item rethinking the co-design of secure aggregation and PIR in this setting, through a new coding technique we introduce to lower the communication costs.\end{enumerate*} 
\end{example}

We will see that the first approach, while solving the problem, incurs a high communication cost in the sharing phase. The second approach reduces this cost, yet requires the assumption of unbounded computation, i.e., $\rep=\nclients$. Since the method of \cref{ex:mceliece} does not directly generalize to the case where $\rep<\nclients$, we construct a scheme in \cref{sec:scheme} that build on the concepts from \cref{ex:mceliece} and design a method for arbitrary incident matrices $\incident$ with $\rep \leq \nobjectives$ that alleviate the interferences resulting from an arbitrary choice of $\incident$ by leveraging dual properties of Reed-Solomon codes and a careful choice of the evaluation points. \rev{\cref{fig:block_diagram} shows the three-stage concept on a high level, and \cref{fig:framework} summarizes the functionality of the sharing and the query stage.}

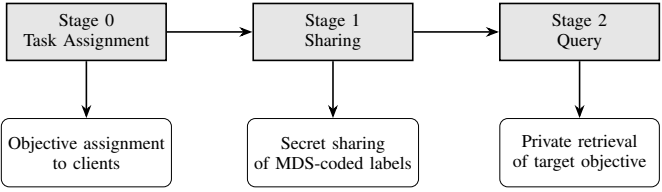
\begin{figure}
    \centering
    \resizebox{\linewidth}{!}{\begin{tikzpicture}[
    node distance=1cm and 1.5cm,
    every node/.style={font=\small},
    box/.style={rectangle, draw, rounded corners, align=center, minimum height=1.2cm, minimum width=2.8cm},
    arrow/.style={-{Stealth}, thick},
    stage/.style={rectangle, draw, thick, align=center, fill=gray!20, minimum width=2.8cm, minimum height=1cm}
]

\node[stage] (stage0) {Stage 0\\Task Assignment};
\node[stage, right=of stage0] (stage1) {Stage 1\\Sharing};
\node[stage, right=of stage1] (stage2) {Stage 2\\Query};

\node[box, below=1cm of stage0] (task) {Objective assignment\\to clients};
\node[box, below=1cm of stage1] (share) {Secret sharing\\ of MDS-coded labels};
\node[box, below=1cm of stage2] (query) {Private retrieval\\ of target objective};

\draw[arrow] (stage0) -- (stage1);
\draw[arrow] (stage1) -- (stage2);

\draw[arrow] (stage0) -- (task);
\draw[arrow] (stage1) -- (share);
\draw[arrow] (stage2) -- (query);

\end{tikzpicture}}
    \caption{\rev{High-level description of a three-stage protocol that first establishes in two stages (assignment and sharing) an MDS-coded data storage pattern based on secret sharing that encodes the aggregated clients computation results for all objectives, and then queries the result for the objective of interest.} \vspace{-.5cm}}
    \label{fig:block_diagram}
\end{figure}

\begin{approach}[Simplified Solution for Shamir Secret Sharing] \label{ex:shamir}
The task assignment stage is trivial as all clients compute the output of all objectives. In the sharing stage, each client $\client$ constructs for each sample $\labelidx$ and each objective $\objective$ a secret sharing 
\begin{equation*}
    \ssclientl = \logit[\labelidx] + x \rlogitsl,
\end{equation*}
encoding the private label $\logit \in \F_q^{\nclasses}$ (encoded by a one-hot encoding into a vector of length $\nclasses$, the number of classes) using $\ncolluders=1$ term $\rlogitsl$ of the size of the label, chosen independently and uniformly at random from $\F_q^{\nclasses}$. Each client $\clientidxb$ then sends the evaluation (share) $\ssclientbl[\alpha_{\clientidxc}]$ to client $\clientidxc$ and receives the share $\ssclientcl[\alpha_{\clientidxb}]$ for all $\clientidxc\in [\nclients]$ for each label $\labelidx$ and objective $\objective$. Aggregating the received shares, client $\clientidxb$ obtains a share $\sumssclientl[\alpha_{\clientidxb}]$ of the sum-secret sharing $\sumssclientl = \sum_{\client=1}^\nclients \ssclientl$, which can be viewed as a codeword of a Generalized Reed Solomon (GRS) code $\storagecode$ with dimension $\kstorage = 2$ and length $\nclients$.

In the query stage, to privately retrieve the labels of interest (i.e., the $\logit$'s for objective $\objective = \queriedobjective$), we design the queries by the following polynomial:
\begin{equation*}
    \ssqueryl = \querymessagel[\labelidx] + x \rqueryl,
\end{equation*}
where $\querymessagel = \begin{cases}
    \mathbf{1} & \text{ if } \objective = \queriedobjective \\
    0 & \text{ otherwise },
\end{cases}$ and $\rqueryl \in \F_q^\nclasses$ are chosen independently and uniformly at random. Each client $\client$ receives the query polynomial evaluated at $\alpha_\client$, and returns $\answerl = \sum_{\objective = 1}^\nobjectives \sumssclientl[\alpha_\client] \ssqueryl[\alpha_\client]$, which is an evaluation of the degree-$2$ polynomial 
\begin{align*}
    \sumssclientl \ssqueryl \!=\! \sumc \logitpartqueriedl &\!+ x\! \left(\!\sumc \rlogitsqueriedl \!+\! \sum_{\objective=1}^\nobjectives \rqueryl \sumc \logit[\labelidx]\!\right) \\
    &+ x^2 \left(\sum_{\objective=1}^\nobjectives \rqueryl \sumc \rlogitsl\right),
\end{align*}
where vector-products are element-wise. Interpolating this polynomial from any subset of the answers $\{\answerl\}_{\client=1}^\nclients$ of size $3$ reveals the desired aggregate of the labels $\sumc \logitpartqueriedl$. %
The communication cost of this scheme is $\nobjectives \nlabels \nclients (\nclients-1) + 3 \nlabels$ symbols in $\F_q^\nclasses$ since the queries and answers are sent for each of the $\nlabels$ labels. Note that we used the same collusion parameter $\ncolluders$ for storage and query codes, which need not hold true in general.
\end{approach}

\begin{figure}
    \centering
    \resizebox{\linewidth}{!}{\input{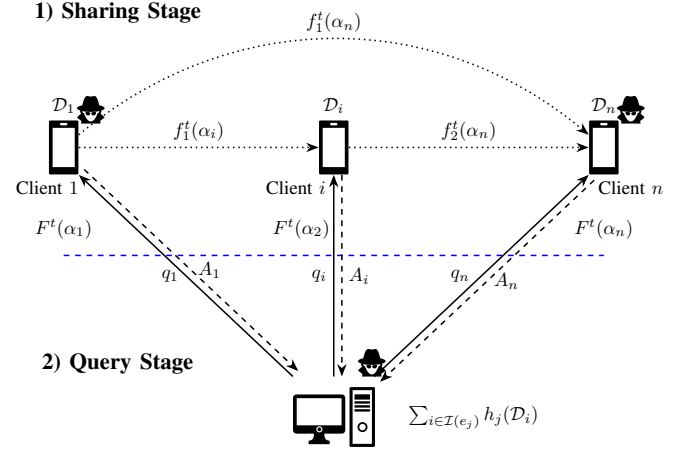}}
    \caption{\rev{Illustration of the sharing and query phase of our protocol.} The objective of interest is hidden from the curious clients. The clients do not learn about the other clients' results. The federator only learns the aggregate clients' results for the objective of interest.  \vspace{-.5cm}}
    \label{fig:framework}
\end{figure}

While the rate of the PIR scheme in \cref{ex:shamir} can be made optimal by using tools from cross-subspace alignment using multiple Shamir secret sharing schemes, the communication cost in the sharing stage is considerably large for the scheme when $\nclients-\ncolluders \gg 2$ and amplified by the number of objectives $\nobjectives$, which gets clear from the above example. 
Next, we illustrate how we further reduce communication costs by using ramp secret sharing and an adapted PIR scheme.

\begin{approach}[Simplified Solution for McEliece Sarwate Secret Sharing] \label{ex:mceliece}
Each client $\client$ constructs for each pair of two labels\footnote{Assuming for simplicity an even number of labels $\nlabels$. This idea can be generalized to jointly encoding an arbitrary number of labels.} $\logit, \logit[\labelidx+1] \in \F_q^{\nclasses}, \labelidx \in \{1, 3, 5, \cdots, \nlabels-1\}$ and for each objective $\objective$ a secret sharing 
\begin{equation*}
    \ssclientl = \logit[\labelidx] + x \logit[\labelidx+1] + x^2 \rlogitsl.
\end{equation*}
Each client $\clientidxb$ then sends the evaluation $\ssclientbl[\alpha_{\clientidxc}]$ to client $\clientidxc$ and receives a share $\ssclientcl[\alpha_{\clientidxb}]$ from each client $\clientidxc\in [\nclients]$ for each label and each objective $\objective$. Each client obtains a share $\sumssclientl[\alpha_{\clientidxb}]$ of the sum-secret sharing $\sumssclientl = \sum_{\client=1}^\nclients \ssclientl$. To privately retrieve the labels of interest, we design the following query polynomial:
\begin{equation*}
    \ssqueryl = \querymessagel[\labelidx] + x^2 \rqueryl,
\end{equation*}
where $\querymessagel = \begin{cases}
    \mathbf{1} & \text{ if } \objective = \queriedobjective \\
    0 & \text{ otherwise },
\end{cases}$ and $\rqueryl \in \F_q^\nclasses$ are chosen independently and uniformly at random. Each client $\client$ receives the query polynomial evaluated at $\alpha_\client$, and returns $\answerl = \sum_{\objective = 1}^\nobjectives \sumssclientl[\alpha_\client] \ssqueryl[\alpha_\client]$, which is an evaluation of the degree-$4$ polynomial 
\begin{align*}
    \sumssclientl \ssqueryl &= \sumc \logitpartqueriedl + x \sumc \logitpartqueriedl[\labelidx+1] \\
    &\hspace{-.3cm}+ x^2\left(\sumc \rlogitsqueriedl[1] + \sumo \rqueryl \sumc \logit \right) \\
    &\hspace{-.3cm}+ x^3 \sumo \rqueryl \sumc \logit[\labelidx+1] + x^4 \sumo \rqueryl \sumc \rlogitsl. 
\end{align*}
Interpolating this polynomial from the set of any $5$ answers in $\{\answerl\}_{\client=1}^\nclients$ reveals $\sumc \logitpartqueriedl$ and $\sumc \logitpartqueriedl[\labelidx+1]$. 
Hence, the communication cost is $\nobjectives \frac{\nlabels}{2} \nclients (\nclients-1) + \frac{5 \nlabels}{2}$ symbols in $\F_q^\nclasses$.
\end{approach}

\cref{ex:mceliece} does not suffer from the drawback of high communication costs in the sharing stage. However, it cannot be directly applied to the graph-based setting. A careful design of a suitable PIR scheme for arbitrary incident matrices $\incident$ will be needed.
Requiring each client to compute the label for all samples and all objective functions exhibits a good utility of the resulting model at the federator, but incurs large computation costs. Hence, a trade-off between the computation complexity and the utility arises in this setting. To leverage this trade-off, we seek solutions that allow reduced computation and simultaneously balance the communication costs in the sharing and the query phase, requiring the design of a specifically tailored graph-based PIR scheme to generalize the scheme in~\cref{ex:mceliece}.

\begin{remark} \label{remark:one_shot_fl}
    We resort to ideas from federated knowledge distillation due to the difficulty in hiding the objective in classical FL settings. Classical FL~\cite{mcmahan2017communication} is an iterative process where at each iteration, the clients compute a gradient based on their individual data and the current global model. They report the result to the federator, who aggregates the received gradients and updates the global model accordingly. After synchronizing the clients with the latest model, the process repeats. While gradient computation can be seen as an arbitrary non-linear function computation, hiding the objective of the federator in an iterative scheme is difficult. Even if the objective is completely hidden from the clients at the time of computing the gradients, the clients are potentially able to infer the federator's objective by simply observing consecutive global model updates. The idea is similar to that of model inversion attacks, which allows the federator to obtain information about the clients' private data after seeing their model updates (gradients). Further, incorporating our framework into an iterative process would incur extensive communication cost. 
\end{remark}

\section{General Scheme} \label{sec:scheme}

We present the end-to-end one-shot FL scheme that preserves the privacy of the federator's objective of interest and the clients' data\rev{, made of three stages: the task allocation, the sharing and the query stage.} We present the scheme for the case that the function results are the labels of a public dataset. Nonetheless, our solution applies for any function that is additively separable.
\paragraph{\rev{Task Assignment Stage}}
The task assignment stage requires each client $\client$ to compute a set of objectives denoted by $\incidentvertex$, i.e., the client computes a function $\function(\data)$, or trains a model $\model{\objective}{\client}$, for each objective $\objective \in \incidentvertex$. For each model, the client creates (predicts) a label out of $\nclasses$ possible classes for each sample of the public unlabeled dataset, obtaining $\nlabels$ labels $\logit, \labelidx \in [\nlabels]$ for each objective $\objective \in \incidentvertex$. The labels of the public dataset can contain either hard or soft information, i.e., a one-hot encoded vector of the class labels or a vector encoding the class probabilities with finite precision. The precision will affect the finite field size $q$ required to store all information necessary without creating finite field overflows in the aggregation process of the clients, i.e., with $\gamma$ quantization steps, we require $q \geq (\gamma-1) \nclients$. Having quantized the labels $\logit$ into vectors from $\F_q^\nclasses$ \rev{enables information-theoretic privacy guarantees through secret sharing.}

\paragraph{Sharing Stage}
We dedicate an evaluation point $\alpha_\client$ for each client $\client \in [\nclients]$, where $\evalc = \eval^\client, \client \in [\nclients]$ for a generator element $\eval$ of the field $\F_q$, $q \geq \max\{\rep + \kstorage-\ncolluderss, (\gamma-1)\nclients\}$, where $\kstorage = \frac{\rep-\ncolludersq +\ncolluderss+1}{2}$ is the dimension of the storage code. 
The storage code is constructed by Multi-Party Computation between the clients. For each objective $\objective \in [\nobjectives]$, each client $\client \in \incidentedge$ splits the set of all $\nlabels$ labels $\logit \in \F_{q}^{\nclasses}$ into $\npartitions=\frac{\nlabels}{\kstorage-\ncolluderss}$ partitions\footnote{Assuming for simplicity that $(\kstorage-\ncolluders) \vert \nlabels$.} $\partitionset_1,\dots,\partitionset_\npartitions$ each of size $\kstorage-\ncolluderss$, where the labels of partition $\partition$ are referred to by $\logitpartentry$ for $\partidx \in [\kstorage-\ncolluderss]$. For each objective $\objective$, each client $\client \in \incidentedge$ encodes each partition $\partition \in [\npartitions]$ into a secret sharing as
\begin{equation}
    \ssclient = \sum_{\partidx=1}^{\kstorage-\ncolluderss} x^{\partidx-1} \logitpartentry + \sum_{\tau=1}^{\ncolluderss} x^{\kstorage-\ncolluderss+\tau-1} \rlogits[\tau], \label{eq:encoding_poly}
\end{equation}
where $\forall \tau\in [\kstorage-\ncolluderss]$ and $\partition \in [\npartitions]$ the vectors $\rlogits[\tau] \in \F_{\nclients}^\nclasses$ are uniformly chosen from $\F_{q}^\nclasses$.
Each client $\client$ sends the evaluation $\ssclient[\alpha_{\clientidxb}]$ to client $\clientidxb$. Each client $\clientidxb$ aggregates all received contributions 
\begin{align*}
    \sumssclient[\alpha_{\clientidxb}] \! &\define \! \sum_{\client\in \incidentedge} \ssclient[\alpha_{\clientidxb}] \! \\
    &= \sum_{\partidx=1}^{\kstorage-\ncolluderss} \!\!\! x^{\partidx\!-\!1} \! \sum_{\client \in \incidentedge} \logitpartentry + \sum_{\tau=1}^{\ncolluderss} x^{\kstorage-\ncolluderss+\tau-1} \rlogitssum[\tau],
\end{align*}
where $\rlogitssum[\tau] \define \sumc \rlogits[\tau]$, 
thus obtains a codeword from an $(\nclients, \kstorage)$-GRS. %
This procedure incurs a total communication cost of $\frac{\nobjectives \nlabels \nclasses \cdot \rep (\rep-1)}{\kstorage-\ncolluderss}$ symbols in $\F_{q}$ given that\footnote{This is more general than the description above, where we would require that $(\kstorage - \ncolluderss) \vert \nlabels$ instead $(\kstorage - \ncolluderss) \vert \nlabels \nclasses$, but is achievable by splitting into fractions of the labels.} $(\kstorage - \ncolluderss) \vert \nlabels \nclasses$. The rate for each secret sharing is $\frac{\kstorage-\ncolluderss}{\rep (\rep-1)}$. Since the federator is interested in only one objective $\queriedobjective \in [\nobjectives]$, the overall sharing rate is $\frac{\kstorage-\ncolluderss}{\nobjectives \rep (\rep-1)}$.

\paragraph{Query Stage}
Having in place a $(\rep, \kstorage)$-GRS storage code (in fact, we have $\nobjectives \npartitions$ of such codes where each information symbol is of size $\frac{\nlabels \nclasses}{\npartitions (\kstorage-\ncolluderss)}$ symbols in $\F_{q}$, we design an (S)PIR scheme that allows the federator to retrieve the labels corresponding to the objective $\objective$ of interest without revealing its identity. 

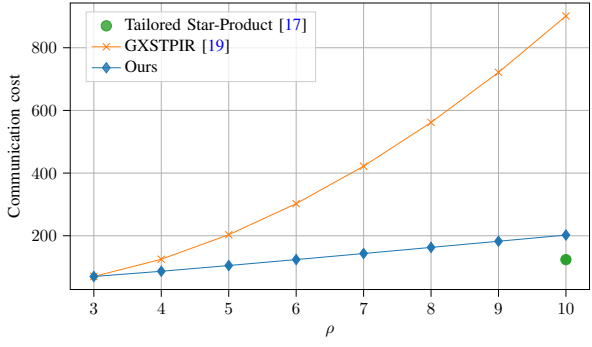
\begin{figure}[!t]
    \centering
    \resizebox{.9\linewidth}{!}{
    \begin{tikzpicture}

\definecolor{darkgray176}{RGB}{176,176,176}
\definecolor{darkorange25512714}{RGB}{255,127,14}
\definecolor{forestgreen4416044}{RGB}{44,160,44}
\definecolor{lightgray204}{RGB}{204,204,204}
\definecolor{steelblue31119180}{RGB}{31,119,180}

\begin{axis}[
legend cell align={left},
legend style={
  fill opacity=0.8,
  draw opacity=1,
  text opacity=1,
  at={(0.03,0.97)},
  anchor=north west,
  draw=lightgray204
},
tick align=outside,
tick pos=left,
x grid style={darkgray176},
xlabel={\(\displaystyle \rho\)},
xmajorgrids,
xmin=2.65, xmax=10.35,
xtick style={color=black},
y grid style={darkgray176},
ylabel={Communication cost},
ymajorgrids,
ymin=28.4375, ymax=942.8125,
ytick style={color=black}
]
\addplot [only marks, semithick, forestgreen4416044, mark=*, mark size=3, mark options={solid, fill=forestgreen4416044}]
table {%
10 123.75
};
\addlegendentry{Tailored Star-Product \cite{freij2017private}}
\addplot [semithick, darkorange25512714, mark=x, mark size=3, mark options={solid}]
table {%
3 70
4 125
5 203.333333333333
6 302.5
7 422
8 561.666666666667
9 721.428571428572
10 901.25
};
\addlegendentry{GXSTPIR \cite{jia2020asymptotic}}
\addplot [semithick, steelblue31119180, mark=diamond*, mark size=3, mark options={solid, fill = steelblue31119180}]
table {%
3 70
4 86.6666666666667
5 105
6 124
7 143.333333333333
8 162.857142857143
9 182.5
10 202.222222222222
};
\addlegendentry{Ours}
\end{axis}

\end{tikzpicture}}
    \caption{Communication Cost in $\nlabels \nclasses$ symbols in $\F_q$ compared to our three-stage protocol paired with GXSTPIR \cite{jia2020asymptotic} and the Star-Product scheme with optimized storage code dimension $\kstorage^\star$ according to \cref{lemma:optimal_comm}. The latter is limited to $\rep = \nclients$, i.e., to non-graph-based settings. The parameters are chosen as $\nclients=10$, $\nobjectives=10$, and $\ncolluderss=\ncolludersq=1$. \vspace{-.5cm}}
    \label{fig:communication_cost}
\end{figure}

Consider for all $\objective \in [\nobjectives], \partition \in [\npartitions]$ the following query polynomial
\begin{equation}
    \ssquery = \querymessage[\partidx] + \sum_{\tau=1}^{\ncolludersq} x^{\kstorage-\ncolluderss+\tau-1} \rquery[\tau], \label{eq:query_poly}
\end{equation}
where $\querymessage = \begin{cases}
    \mathbf{1} & \text{ if } \objective = \queriedobjective \\
    0 & \text{ otherwise },
\end{cases}$ and $\rquery[\tau] \in \F_q^\nclasses$ are chosen independently and uniformly at random. For each $\objective \in \nobjectives$, client $\client \in \incidentedge$ receives an evaluation $\ssquery[\alpha_{\client}]$ of this polynomial, which it multiplies (element-wise) by all $\sumssclient[\alpha_{\client}]$ for $\partition \in [\npartitions]$, and additionally by 
\begin{equation*}
    \dualco \define \left(\prod_{\clientidxb \in \incidentedge \setminus \{\client\}} (\evalc - \evalc[\clientidxb])\right)^{-1}.
\end{equation*}
The choice of $\dualco$ is justified by the duals of GRS codes and will allow for arbitrary task assignment patterns. 
The federator receives the answers $\answer = \{\answerpart[1], \cdots, \answerpart[\npartitions]\}$, where $\answerpart$ are evaluations of the following answer polynomial
\begin{align*}
    &\ssanswerpart = \sumi \dualco \sumssclient[x] \ssquery[x] \\
    &= \dualcoqueried \!\! \sum_{\partidx=1}^{\kstorage-\ncolluderss} x^{\partidx-1} \!\!\!\! \sum_{\client\in \incidentedge[\queriedobjective]} \!\!\! \logitpartentryqueried + \sumi \dualco \!\!\!\!\! \sum_{\tau=1}^{\kstorage+\ncolludersq-1} \!\!\!\!\! x^{\kstorage-\ncolluderss+\tau-1} \ranswer[\tau],
\end{align*}
and $\forall \tau \in [\kstorage+\ncolludersq-1]$, $\ranswer[\tau]$ are interference terms being potentially different for each client. The multiplication $\sumssclient[x] \ssquery[x]$ is done element-wise.
\paragraph{Reconstruction Stage}
For $\evalidx \in [\kstorage-\ncolluderss]$, the federator sums over all clients' answers to obtain
\begin{align*}
    \evalsum &\define \sumc \evalc^{-\evalidx} \answerpart \\
    &= \sum_{\partidx=1}^{\kstorage-\ncolluderss} \left(\sum_{\client\in \incidentedge[\queriedobjective]} \logitpartentryqueried\right) \sumioqueried \dualcoqueried \evalc^{\partidx-\evalidx-1}.
\end{align*}
Let $\logitpartqueriedsum \define \sum_{\client\in \incidentedge[\queriedobjective]} \logitpartentryqueried$. By computing $\evalsum$ for all $\evalidx \in [\kstorage-\ncolluderss]$, the federator can obtain the desired information as
\begin{align*}
    \begin{pmatrix}
        \logitpartqueriedsum[1],
        \logitpartqueriedsum[2],
        \cdots,
        \logitpartqueriedsum[\kstorage-\ncolluderss]
    \end{pmatrix}^{T} \!\! = \mathbf{P}^{-1} \begin{pmatrix}
        \evalsum[1],
        \evalsum[2],
        \cdots,
        \evalsum[\kstorage-\ncolluderss]
    \end{pmatrix}^{T}\!\!\!\!,
\end{align*}
where $\mathbf{P}$ is the following invertible matrix
\begin{align*}
\mathbf{P} &= \sumioqueried \begin{pmatrix}
    \dualcoqueried \evalc^{-1} & 0 & \cdots & 0 & 0 \\
    \dualcoqueried \evalc^{-2} & \dualcoqueried \evalc^{-1} & 0 & \cdots & 0 \\
    \vdots & \ddots & \ddots & \ddots & 0 \\
    \dualcoqueried \evalc^{-\kstorage+\ncolluderss} & \cdots & \cdots & \cdots & \dualcoqueried \evalc^{-1}
\end{pmatrix}.
\end{align*}

\begin{table*}[t!]
\caption{Comparison of Secret Sharing and PIR Rates, and Communication Costs. The parameter $\kstorage^\star$ is given by \cref{lemma:optimal_comm}.}
\centering
\begin{tabular}{l|l|l|l}
\toprule
Method & Sharing Rate & PIR Rate & \rev{Total Communication Cost in $\mathbb{F}_q$} \\ \midrule
Ours & 
$\frac{\rep-\ncolluderss-\ncolludersq+1}{2\nobjectives \rep (\rep-1)}$ & 
$\frac{\rep-\ncolludersq-\ncolluderss+1}{2\nclients}$ & 
$\frac{2\nobjectives \nlabels \nclasses \cdot \rep (\rep-1)}{\rep-\ncolluderss-\ncolludersq+1} + \frac{2 \nlabels \nclasses \cdot \nclients}{\rep-\ncolludersq-\ncolluderss+1}$ \\\hline %

GXSPIR \cite{jia2020asymptotic} & 
$\frac{1}{\nobjectives \rep(\rep-1)}$ & 
$\frac{\rep-\ncolludersq-\ncolluderss}{\nclients}$ & 
$\nobjectives \nlabels \nclasses \cdot \rep (\rep-1) + \frac{\nlabels \nclasses \cdot \nclients}{\rep-\ncolluderss-\ncolludersq}$ \\\hline

$\rep=\nclients$ (modified \cite{freij2017private}) & 
$\frac{\kstorage^\star-\ncolluderss}{\nobjectives \nclients (\nclients-1)}$ & 
$\frac{(\kstorage^\star-\ncolluders)(\nclients-\kstorage^\star-\ncolludersq+1)}{\nclients \kstorage^\star}$ & 
$\frac{\nobjectives \nlabels \nclasses \cdot \nclients (\nclients-1)}{\kstorage^\star-\ncolluderss} + \frac{\nlabels \nclasses}{\kstorage^\star-\ncolluderss} \frac{\kstorage^\star \nclients}{\nclients - \kstorage^\star - \ncolludersq +1}$ \\ \bottomrule
\end{tabular}
\label{tab:rates}
\end{table*}

Having obtained all aggregated labels for the objective $\queriedobjective$ of interest, the federator retrains a suitable model leveraging the public dataset and the heterogeneity of the clients' data.

\begin{figure}[!t]
    \centering
    \resizebox{.9\linewidth}{!}{
    \begin{tikzpicture}

\definecolor{darkgray176}{RGB}{176,176,176}
\definecolor{darkorange25512714}{RGB}{255,127,14}
\definecolor{forestgreen4416044}{RGB}{44,160,44}
\definecolor{lightgray204}{RGB}{204,204,204}
\definecolor{steelblue31119180}{RGB}{31,119,180}

\begin{axis}[
legend cell align={left},
legend style={
  fill opacity=0.8,
  draw opacity=1,
  text opacity=1,
  at={(0.03,0.97)},
  anchor=north west,
  draw=lightgray204
},
log basis y={10},
tick align=outside,
tick pos=left,
x grid style={darkgray176},
xlabel={\(\displaystyle \rho\)},
xmajorgrids,
xmin=6.55, xmax=104.45,
xtick style={color=black},
y grid style={darkgray176},
ylabel={Communication cost},
ymajorgrids,
ymin=1078.38386960765, ymax=253788.603034318,
ymode=log,
ytick style={color=black}
]
\addplot [only marks, semithick, forestgreen4416044, mark=*, mark size=3, mark options={solid, fill=forestgreen4416044}]
table {%
100 2305.55555555556
};
\addlegendentry{Tailored Star-Product \cite{freij2017private}}
\addplot [semithick, darkorange25512714, mark=x, mark size=3, mark options={solid}]
table {%
11 2300
12 2690
13 3153.33333333333
14 3665
15 4220
16 4816.66666666667
17 5454.28571428571
18 6132.5
19 6851.11111111111
20 7610
21 8409.09090909091
22 9248.33333333333
23 10127.6923076923
24 11047.1428571429
25 12006.6666666667
26 13006.25
27 14045.8823529412
28 15125.5555555556
29 16245.2631578947
30 17405
31 18604.7619047619
32 19844.5454545455
33 21124.347826087
34 22444.1666666667
35 23804
36 25203.8461538462
37 26643.7037037037
38 28123.5714285714
39 29643.4482758621
40 31203.3333333333
41 32803.2258064516
42 34443.125
43 36123.0303030303
44 37842.9411764706
45 39602.8571428571
46 41402.7777777778
47 43242.7027027027
48 45122.6315789474
49 47042.5641025641
50 49002.5
51 51002.4390243902
52 53042.380952381
53 55122.3255813954
54 57242.2727272727
55 59402.2222222222
56 61602.1739130435
57 63842.1276595745
58 66122.0833333333
59 68442.0408163265
60 70802
61 73201.9607843137
62 75641.9230769231
63 78121.8867924528
64 80641.8518518519
65 83201.8181818182
66 85801.7857142857
67 88441.7543859649
68 91121.724137931
69 93841.6949152542
70 96601.6666666667
71 99401.6393442623
72 102241.612903226
73 105121.587301587
74 108041.5625
75 111001.538461538
76 114001.515151515
77 117041.492537313
78 120121.470588235
79 123241.449275362
80 126401.428571429
81 129601.408450704
82 132841.388888889
83 136121.369863014
84 139441.351351351
85 142801.333333333
86 146201.315789474
87 149641.298701299
88 153121.282051282
89 156641.265822785
90 160201.25
91 163801.234567901
92 167441.219512195
93 171121.204819277
94 174841.190476191
95 178601.176470588
96 182401.162790698
97 186241.149425287
98 190121.136363636
99 194041.123595506
100 198001.111111111
};
\addlegendentry{GXSTPIR \cite{jia2020asymptotic}}
\addplot [semithick, steelblue31119180, mark=diamond*, mark size=3, mark options={solid, fill = steelblue31119180}]
table {%
11 2300
12 1826.66666666667
13 1610
14 1496
15 1433.33333333333
16 1400
17 1385
18 1382.22222222222
19 1388
20 1400
21 1416.66666666667
22 1436.92307692308
23 1460
24 1485.33333333333
25 1512.5
26 1541.17647058824
27 1571.11111111111
28 1602.10526315789
29 1634
30 1666.66666666667
31 1700
32 1733.91304347826
33 1768.33333333333
34 1803.2
35 1838.46153846154
36 1874.07407407407
37 1910
38 1946.20689655172
39 1982.66666666667
40 2019.35483870968
41 2056.25
42 2093.33333333333
43 2130.58823529412
44 2168
45 2205.55555555556
46 2243.24324324324
47 2281.05263157895
48 2318.97435897436
49 2357
50 2395.12195121951
51 2433.33333333333
52 2471.62790697674
53 2510
54 2548.44444444444
55 2586.95652173913
56 2625.53191489362
57 2664.16666666667
58 2702.85714285714
59 2741.6
60 2780.39215686275
61 2819.23076923077
62 2858.11320754717
63 2897.03703703704
64 2936
65 2975
66 3014.0350877193
67 3053.10344827586
68 3092.20338983051
69 3131.33333333333
70 3170.49180327869
71 3209.67741935484
72 3248.88888888889
73 3288.125
74 3327.38461538461
75 3366.66666666667
76 3405.97014925373
77 3445.29411764706
78 3484.63768115942
79 3524
80 3563.38028169014
81 3602.77777777778
82 3642.19178082192
83 3681.62162162162
84 3721.06666666667
85 3760.52631578947
86 3800
87 3839.48717948718
88 3878.98734177215
89 3918.5
90 3958.02469135802
91 3997.56097560976
92 4037.10843373494
93 4076.66666666667
94 4116.23529411765
95 4155.81395348837
96 4195.40229885057
97 4235
98 4274.60674157303
99 4314.22222222222
100 4353.84615384615
};
\addlegendentry{Ours}
\end{axis}

\end{tikzpicture}}
    \caption{\rev{Communication Cost in $\nlabels \nclasses$ symbols in $\F_q$ compared to our three-stage protocol paired with GXSTPIR \cite{jia2020asymptotic} and the Star-Product scheme with optimized storage code dimension $\kstorage^\star$ according to \cref{lemma:optimal_comm}. The latter is limited to $\rep = \nclients $, i.e., to non-graph-based settings. The parameters are chosen as $\nclients=100$, $\nobjectives=20$, and $\ncolluderss=\ncolludersq=5$.} \vspace{-.5cm}}
    \label{fig:communication_cost_100}
\end{figure}

\paragraph{Properties of the Scheme} We state in the following the most important properties of our proposed scheme, which is the sharing rate, the PIR rate, and the privacy guarantee.

\begin{proposition}[Sharing Rate]
    The rate of the proposed sharing scheme is
    \begin{equation*}
        {R}_{\mathrm{share}} = \frac{\rep-\ncolluderss-\ncolludersq+1}{2\nobjectives \rep (\rep-1)}.
    \end{equation*}
\end{proposition}
\begin{proof}
    We assume in the worst case that all clients' function results are independent and uniformly distributed. By the above choice of $\kstorage$, for each objective the number of labels contained in each secret sharing according to \eqref{eq:encoding_poly} is $\frac{\rep-\ncolluderss-\ncolludersq+1}{2}$, and the number of shares (of the same size) transmitted is given by $\rep (\rep-1)$. Since the federator is only interested in one out of all $\nobjectives$ objectives, the rate deteriorates by $\nobjectives$.
\end{proof}

\begin{theorem}[PIR Rate]
    The rate of the proposed PIR scheme is
    \begin{align*}
       {R}_{\mathrm{PIR}} = \frac{\rep-\ncolludersq-\ncolluderss+1}{2\nclients} \stackrel{(\star)}{=} \frac{\rep-\kstorage-\ncolludersq+1}{\nclients}
    \end{align*}
    where $(\star)$ holds for $\rep \leq 2\kstorage + \ncolludersq-1$ and $\ncolluderss = 2\kstorage - \rep + \ncolludersq -1$.
\end{theorem}
\begin{proof}
The rate of the PIR scheme for each $\edge \in \edges$ is $\frac{\kstorage-\ncolluderss}{2\kstorage + \ncolludersq - \ncolluderss -1}$. Setting $\rep = 2\kstorage + \ncolludersq - \ncolluderss -1$, we obtain for $\rep \leq 2\kstorage + \ncolludersq-1$ a per-objective rate of $\frac{\rep-\kstorage-\ncolludersq+1}{\rep}$, where $\ncolluderss = 2\kstorage - \rep + \ncolludersq -1$. Since we also need to query clients $\client \in [\nclients] \setminus \incidentedge[\queriedobjective]$ for reasons of privacy, the overall rate of the PIR scheme is $\frac{\rep-\kstorage-\ncolludersq+1}{\nclients}$. 
\end{proof}
The  communication cost is, hence, $\frac{\nclients}{\rep-\kstorage-\ncolludersq+1}$. 
Let the set of all messages received by client $\clientidxb$ be $\allreceived[\client] \define \{\ssclient[\evalc[\clientidxb]]\}_{\client \in \incidentedge \setminus \{\clientidxb\}, \objective \in \incidentvertex[\clientidxb]}$. Further, let the set of all messages received by all clients in $\mathcal{T}\subset[\nclients]$ be $\allreceivedcol[\mathcal{T}] \define \{\allreceived[\client]\}_{\client \in \mathcal{T}}$.

\begin{theorem}[Privacy from Clients and Objective-Hiding] \label{thm:privacy}
    The clients' computation results are private against any set of $\ncolluderss$ clients (cf. \cref{def:privacy_clients}). The objective $\queriedobjective$ is private against any set of $\ncolludersq$ clients (cf. \cref{def:objective_hiding}), i.e., for any two sets of clients $\colss, \colquery \subset [\nclients]\setminus\{\client\}, \vert \colss \vert \leq \ncolluderss, \vert \colquery \vert \leq \ncolludersq$, we have
    \begin{align*}
        \condmutinf{\!\{\logitrv\}_{\objective \in \incidentvertex, \labelidx \in [\nlabels]}}{\allreceivedcol[\colss]\!\!}{\!\!\{\logitrv\}_{\objective \in \incidentvertex, \client \in \colss, \labelidx\in [\nlabels]}\!} &= 0, \forall \client \! \in \! [\nclients],\\
        \mutinf{\{\queryrvcol\}, \allcol}{\queriedobjective} &= 0.
    \end{align*}    
\end{theorem}

\begin{theorem}[Correctness] \label{thm:correctness} The sum of labels of interest $\sumioqueried \logitqueried$, $\labelidx\in[\nlabels]$, is decodable from the answers, i.e., 
    \begin{align*}
        \entropy{\!\bigg\{\!\sum_{\client \in \incidentedge[\queriedobjective]} \!\! \logitpartqueriedrv\bigg\}_{\partition \in [\npartitions], \partidx\in [\kstorage-\ncolluderss]} \hspace{-.5cm} \mid \{\answerpart\}_{\partition \in [\npartitions], \client\in \incidentedge[\queriedobjective]}\!} = 0.
    \end{align*}
\end{theorem}

We compare the communication cost of our scheme in \cref{fig:communication_cost,fig:communication_cost_100} to the application of two PIR schemes from the literature \rev{for $\nclients=10$ and $\nclients=100$, respectively,} and provide the corresponding rates in \cref{tab:rates}. The details of the comparison are deferred to \cref{sec:comparison}.

\section{Extension to Privacy from the Federator}
While the above scheme complies with the privacy notion according to \cref{def:privacy_clients}, privacy from the federator as in \cref{def:privacy_federator} is not ensured. Therefore, recall the answer polynomial $\ssanswerpart$ above. The crucial aspect is that the interference terms $\ranswer[\tau]$ contain sensitive information about $\logitpart$ for $\partidx\in[\kstorage-\ncolluderss]$ and $\objective \neq \queriedobjective$, i.e., they depend on the function results of the clients, thereby leaking potential information about clients' results beyond the objective $\queriedobjective$ of interest.
To ensure user-side privacy, where the federator does not learn anything about the clients models beyond the objective of interest, we assume the existence of shared randomness among all clients unknown to the federator. Leveraging this shared randomness, the clients construct a randomized polynomial
\begin{align*}
    \ssrandpart = \sum_{\tau=1}^{\kstorage+\ncolludersq-1} x^{\kstorage-\ncolluderss+\tau-1} \rrand[\tau],
\end{align*}
where $\{\rrand[\tau]\}_{\tau \in [\kstorage+\ncolludersq-1]}$ is known to all clients, but unknown to the federator. For each partition $\partition \in [\npartitions]$, each client $\client$ replies to the queries $\ssquery, \objective \in \incidentvertex$ with the answer
\begin{align*}
    \answerpartsym = \sumo \dualco \sumssclient[\alpha_{\client}] \ssquery[\alpha_{\client}] + \randpart,
\end{align*}
which is an evaluation of the \emph{re-randomized} polynomial
\begin{align*}
    \ssanswerpartsym &= %
    \dualcoqueried \sum_{\partidx=1}^{\kstorage-\ncolluderss} x^{\partidx-1} \sum_{\client\in \incidentedge[\queriedobjective]} \logitpartentryqueried \\
    &\hspace{.8cm}+ \sum_{\tau=1}^{\kstorage+\ncolludersq-1} x^{\kstorage-\ncolluderss+\tau-1} \bigg(\rrand[\tau] + \sumi \dualco \ranswer[\tau]\bigg),
\end{align*}
which is, by the one-time pad, guarantees the privacy of the $\ranswer[\tau]$ that contain potentially sensitive information about the clients' computation beyond the objective of interest. The recovery process as in \cref{sec:scheme} remains unchanged. Let the answer $\answer$ received from client $\client$ be $\answer \define \{\answerpartsym\}_{\partition \in [\npartitions]}$, then we have the following privacy statement.

\begin{theorem}[Symmetric Privacy] \label{thm:sym_privacy}
    In addition to the privacy satisfied according to \cref{thm:privacy}, the federator learns nothing beyond the aggregation of $\rep$ clients' predictions for the objective $\queriedobjective$ of interest. Formally,
    \begin{align*}
        \condmutinf{\answer[{[\nclients]}], \{\queryrv[{[\nclients]}]\}}{\{\logitrv\}}{\bigg\{\sumioqueried \logitqueriedrv\bigg\}_{\labelidx \in [\nlabels]}} = 0.
    \end{align*}
    where for clarity we define $\{\queryrv[{[\nclients]}]\} \define \{\queryrv[{[\nclients]}]\}_{\objective \in [\nobjectives], \partition \in [\npartitions]}$ and $\{\logitrv\} \define \{\logitrv\}_{\objective \in \incidentvertex, \client \in [\nclients], \labelidx\in [\nlabels]}$.
\end{theorem}

\begin{proof}
    By the one-time pad, all interference terms $\sumi \dualco \ranswer[\tau]$ containing sensitive information are perfectly hidden from the federator, and hence, all labels $\{\logitrv\}_{\client \in [\nclients], \objective \in \incidentvertex, \labelidx\in [\nlabels]}$ with the exception of $\{\sumioqueried \logitqueriedrv\}_{\labelidx \in [\nlabels]}$ are statistically independent from the answers $\answer[{[\nclients]}]$. Further, the queries $\{\queryrv[{[\nclients]}]\}_{\objective \in [\nobjectives], \partition \in [\npartitions]}$ are independent of the labels. This holds even if all clients' polynomials could be exactly reconstructed by the federator.
\end{proof}

\section{Comparison to Cross-Subspace-Alignment and Star-Product Codes} \label{sec:comparison}

Our scheme encapsulates a secret sharing stage and a new PIR scheme for graph-based MDS coded storage patterns specifically tailored to generalized secret sharing schemes. When restricting to the suboptimal Shamir Secret Sharing in the sharing stage, the method of \cite{jia2020asymptotic} can be used instead of our PIR scheme. In a non-graph-based setting (when $\rep = \nclients$), known methods from PIR over MDS coded data apply for generalized secret sharing schemes due to their MDS structure. However, for optimal overall communication costs, we formulate and solve an optimization problem that trades the sharing rate against the PIR rate and finds the optimal operating point. We elaborate on the two extreme cases in the following, and provide a comparison to our scheme.

\begin{figure}[!t]
    \centering
    \resizebox{.9\linewidth}{!}{
    \begin{tikzpicture}

\definecolor{darkgray176}{RGB}{176,176,176}
\definecolor{darkorange25512714}{RGB}{255,127,14}
\definecolor{forestgreen4416044}{RGB}{44,160,44}
\definecolor{steelblue31119180}{RGB}{31,119,180}

\begin{axis}[
tick align=outside,
tick pos=left,
x grid style={darkgray176},
xlabel={\(\displaystyle \rho\)},
xmajorgrids,
xmin=2.65, xmax=10.35,
xtick style={color=black},
y grid style={darkgray176},
ylabel={Sharing rate},
ymajorgrids,
ymin=0.000333333333333333, ymax=0.0174444444444444,
ytick style={color=black}
]
\addplot [semithick, steelblue31119180, dashed, mark=diamond*, mark size=3, mark options={solid}]
table {%
3 0.0166666666666667
4 0.0125
5 0.01
6 0.00833333333333333
7 0.00714285714285714
8 0.00625
9 0.00555555555555555
10 0.005
};
\addplot [semithick, darkorange25512714, dashed, mark=diamond*, mark size=3, mark options={solid}]
table {%
3 0.0166666666666667
4 0.00833333333333333
5 0.005
6 0.00333333333333333
7 0.00238095238095238
8 0.00178571428571429
9 0.00138888888888889
10 0.00111111111111111
};
\addplot [semithick, forestgreen4416044, dashed, mark=diamond*, mark size=3, mark options={solid}]
table {%
10 0.00888888888888889
};
\end{axis}

\end{tikzpicture}}
    \caption{Secret Sharing rates compared to GXSTPIR \cite{jia2020asymptotic} and the Star-Product scheme with optimized storage code dimension $\kstorage^\star$ according to \cref{lemma:optimal_comm}. The latter is limited to $\rep = \nclients$, i.e., to non-graph-based settings. The parameters are chosen as $\nclients=10$, $\nobjectives=10$, and $\ncolluderss=\ncolludersq=1$. \vspace{-.5cm}}
    \label{fig:sharing_rates}
\end{figure}
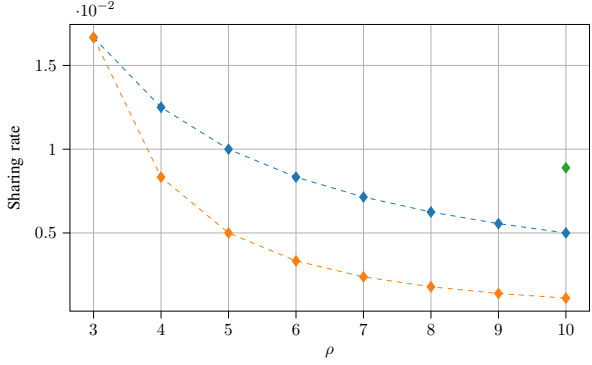

\paragraph{Star-Product PIR \cite{freij2017private}} Considering non-graph-based settings (i.e., $\rep=\nclients$), since secret sharing is a Reed-Solomon code, known results from private information retrieval over MDS-coded data such as those in \cite{freij2017private} can be applied in our framework. %
Applying such results yields an interesting trade-off between the design of the storage code and the query code. This trade-off is not present in our scheme due to our construction being tailored to this setting. Star-product-based PIR schemes consist of a storage cost $\mathcal{C}$ and a query code $\mathcal{D}$. On a high level, each client (or server) returns codewords from the star-product code $\mathcal{C} \star \mathcal{D}$, and the messages of interest are encoded as erasures in the code. The rate of this PIR scheme is $(d_{\mathcal{C} \star \mathcal{D}}-1)/\nclients$, while being private against $\ncolludersq = d_{\mathcal{D}^\perp}-1$ colluders, where $d_{\mathcal{D}^\perp}$ is the minimum distance of the dual code of $\cD$. Given a desired $\ncolludersq$, the dual of $\querycode$ must satisfy $d_{\querycode^\perp} = \ncolludersq+1$ and have dimension $\kquerydual = \nclients-\ncolludersq$. Hence, $\querycode$ is an $(\nclients,\ncolludersq)$-$\grs{}$ code. If we choose the same code locators for both codes, then the star product $\productcode$ is an $(\nclients,\min\{\kstorage+\ncolludersq-1, \nclients\})$-$\grs{}$ code with minimum distance $d_{\productcode} = \nclients - \kstorage - \ncolludersq + 2$ given that $\kstorage+\ncolludersq-1 \leq \nclients$. The PIR scheme achieves a rate of $\frac{d_{\productcode}-1}{\nclients} = \frac{\nclients - \kstorage - \ncolludersq + 1}{\nclients}$ \cite{freij2017private}. Fixing $\ncolludersq$, it is desirable to choose a storage code with a small parameter $\kstorage$. However, in this case the randomness of the storage code is recovered as a by-product, and hence, the rate comes with an additional factor of $\frac{\kstorage-\ncolluderss}{\kstorage}$. The rate results as $\frac{\nclients - \kstorage - \ncolludersq + 1}{\nclients} \frac{\kstorage-\ncolluderss}{\kstorage}$.  %
The download cost of the PIR scheme is given by $\frac{\nlabels \nclasses}{\kstorage-\ncolluderss} \frac{\kstorage \nclients}{\nclients - \kstorage - \ncolludersq +1}$ symbols in $\F_{q}$. However, each of the clients is required to train $\nobjectives$ models for each potential function.

\begin{lemma} \label{lemma:optimal_comm}
    The lowest communication cost is given by
    \begin{align*}
        \frac{\nobjectives \nlabels \nclasses \cdot \nclients (\nclients-1)}{\kstorage^\star-\ncolluderss} + \frac{\nlabels \nclasses}{\kstorage^\star-\ncolluderss} \frac{\kstorage^\star \nclients}{\nclients - \kstorage^\star - \ncolludersq +1}
    \end{align*}
    symbols in $\F_{q}$, where $\kstorage^\star$ is chosen from $\{\lfloor \kstorage^\prime \rfloor, \lceil \kstorage^\prime \rceil\}$ as given below to minimize the above cost. Defining $c_1 \define (\nclients - \ncolludersq + 1)$ and $c_2 \define \nobjectives \nclients (\nclients-1)$, we have
    \begin{align*}
        \kstorage^\prime = \frac{1}{c_2 - \nclients} \left(c_1 c_2 - \sqrt{c_1^2 c_2^2 - (c_2 - \nclients) (c_1^2 c_2 + \nclients c_1 \ncolluderss)}\right).
    \end{align*}
\end{lemma}

\begin{figure}[!t]
    \centering
    \vspace{.27cm}
    \resizebox{.9\linewidth}{!}{
    \begin{tikzpicture}

\definecolor{darkgray176}{RGB}{176,176,176}
\definecolor{darkorange25512714}{RGB}{255,127,14}
\definecolor{forestgreen4416044}{RGB}{44,160,44}
\definecolor{lightgray204}{RGB}{204,204,204}
\definecolor{steelblue31119180}{RGB}{31,119,180}

\begin{axis}[
legend cell align={left},
legend style={
  fill opacity=0.8,
  draw opacity=1,
  text opacity=1,
  at={(0.03,0.97)},
  anchor=north west,
  draw=lightgray204
},
tick align=outside,
tick pos=left,
x grid style={darkgray176},
xlabel={\(\displaystyle \rho\)},
xmajorgrids,
xmin=2.65, xmax=10.35,
xtick style={color=black},
y grid style={darkgray176},
ylabel={PIR rate},
ymajorgrids,
ymin=0.0533333333333333, ymax=0.835555555555556,
ytick style={color=black}
]
\addplot [semithick, steelblue31119180, mark=*, mark size=3, mark options={solid}]
table {%
3 0.1
4 0.15
5 0.2
6 0.25
7 0.3
8 0.35
9 0.4
10 0.45
};
\addlegendentry{Ours}
\addplot [semithick, darkorange25512714, mark=*, mark size=3, mark options={solid}]
table {%
3 0.1
4 0.2
5 0.3
6 0.4
7 0.5
8 0.6
9 0.7
10 0.8
};
\addlegendentry{GXSTPIR \cite{jia2020asymptotic}}
\addplot [semithick, forestgreen4416044, mark=*, mark size=3, mark options={solid}]
table {%
10 0.0888888888888889
};
\addlegendentry{Tailored Star-Product \cite{freij2017private}}
\end{axis}

\end{tikzpicture}}
    \caption{Private information retrieval rates compared to GXSTPIR \cite{jia2020asymptotic} and the Star-Product scheme with optimized storage code dimension $\kstorage^\star$ according to \cref{lemma:optimal_comm}. The latter is limited to $\rep = \nclients$, i.e., to non-graph-based settings. The parameters are chosen as $\nclients=10$, $\nobjectives=10$, and $\ncolluderss=\ncolludersq=1$. \vspace{-.5cm}}
    \label{fig:pir_rates}
\end{figure}
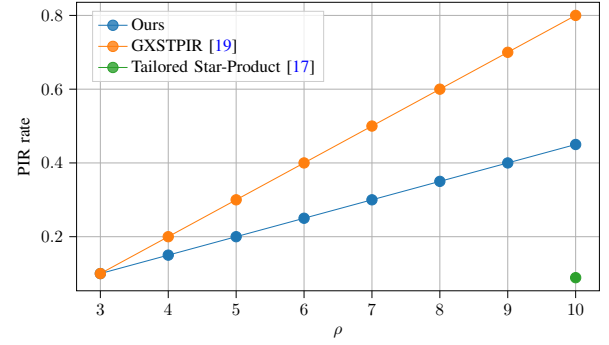

\begin{approach}

    Consider the setting of \cref{ex:description}. We have $\nclients = 5$ clients $\client \in [5]$ and $\nlabels$ labels, split into $\npartitions = \nlabels/2$ partitions of size two. We assume no collusions between clients, i.e., $\ncolluderss=\ncolludersq=\ncolluders=1$. Let for some choice of $\nobjectives$ the optimal code dimension be $\kstorage^\star = 3$. For each $\partition \in [\npartitions]$, each client $\client$ creates shares of the form 
    \begin{equation*}
    \ssclient = \logitpartentry[1] + x \cdot \logitpartentry[2] + x^2 \cdot \rlogits[1],
    \end{equation*}
    and sends to each client $\clientidxb \neq \client$ a share $\ssclient[\alpha_{\clientidxb}]$, keeping the share $\ssclient[\alpha_\client]$ to itself. Summing the shares of all incoming clients, each client $\client$ obtains $\forall \partition \in [\npartitions]$ a share $\sumssclient[\alpha_\client]$ of the polynomial $\sumssclient$, corresponding to a codeword of an $(\nclients,3)$-$\grs{}(\alpha, \mathbf{1}_\nclients)$ code $\storagecode$, where $\alpha = (\alpha_1, \cdots, \alpha_\nclients)$ are the evaluators and $\mathbf{1}_\nclients$ the column multipliers. The communication cost in the sharing stage is $\nobjectives \frac{\nlabels}{2} \nclients (\nclients-1)$ symbols in $\F_{q}^\nclasses$.

    For privacy against $\rev{\ncolludersq} = 1$, we choose a query code $\querycode$ whose dual has dimension $\kquerydual = \nclients - \ncolluders = 4$, hence an $(\nclients,\nclients-\kquerydual)$-$\grs{}(\alpha, \mathbf{1}_\nclients)$. The star product code $\productcode$ then is an $(\nclients,3)$-$\grs{}(\alpha, \mathbf{1}_\nclients)$ code, with minimum distance $d_{\productcode} = 3$. The overall rate of the PIR scheme is $\frac{d_{\productcode}-1}{\nclients} = \frac{2}{5}$. The overall communication cost of the scheme is $\nobjectives \frac{\nlabels}{2} \nclients (\nclients-1) + \frac{5}{2\nlabels}$ symbols in $\F_{q}^\nclasses$, which is the same as for \cref{ex:mceliece}. However, this method does not apply to graph-based settings with arbitrary task assignments due to the interplay of storage and query code, but can yield better results for $\rep=\nclients$ as shown in \cref{fig:communication_cost}.
\end{approach}

\paragraph{Graph-based PIR with Cross-Subspace Alignment} When restricting to Shamir Secret Sharing schemes, the clients would construct a secret sharing for each prediction individually. If all secret sharing instances follow the same construction, cross-subspace alignment can improve the rate of the PIR scheme. In fact, the rate of GXSTPIR was shown to be $\frac{\rep-\ncolluderss-\ncolludersq}{\nclients}$ \cite{jia2020asymptotic}. For XSTPIR (non-graph based), let $\rep = \nclients$. Then we have a rate of $\frac{\rep-\ncolluderss-\ncolludersq}{\rep}$ \cite{jia2019cross}, where Shamir Secret Sharing schemes are shown to be capacity-achieving through cross-subspace alignment. However, being restricted to Shamir Secret Sharing is undesirable in our case since the rate of the sharing stage is $\frac{\kstorage-\ncolluderss}{\rep (\rep-1)}$ for $\rho \leq \kstorage$, which is the worst for $\kstorage-\ncolluderss=1$, where the rate is $\frac{1}{\rep (\rep-1)}$. The optimal rate results when $\rep$ is maximal, which yields $\frac{\rep-\ncolluderss}{\rep (\rep-1)}$. In comparison, with our proposed PIR scheme for arbitrary $\kstorage-\ncolluderss \leq \rep-\ncolluderss$, we obtain a PIR rate of $\frac{\kstorage-\ncolluderss}{2\kstorage + \ncolludersq - \ncolluderss - 1}$ for $\rep \geq 2\kstorage + \ncolludersq - \ncolluderss - 1$, which shows that for non-fixed $\rep$, large $\kstorage-\ncolluderss$ are desirable. Equivalently, we can write $\frac{\rep-\kstorage-\ncolludersq+1}{\rep} = \frac{\rep-\ncolludersq-\ncolluderss+1}{2\rep}$, hence, asymptotically in $\rep$, the rate goes to $\frac{1}{2}$.

\begin{remark}[Cross-Subspace Alignment in our Scheme]
    The authors of \cite{jia2020asymptotic} use cross-subspace alignment by a careful choice of evaluation points for different Shamir Secret Sharing instances to construct capacity-achieving PIR schemes. On the contrary, we jointly design the storage and query code to reduce the overall communication cost dominated by the sharing stage. This leads to occupying all dimensions of the answers with only one secret sharing. If the communication cost in the sharing stage was of lower importance, ideas from cross-subspace alignment could be incorporated into our framework by a deliberate choice of the evaluation points for different McEliece-Sarwate secret sharing instances.
\end{remark}

We compare in \cref{fig:communication_cost} the overall communication cost of our proposed three-stage protocol for the proposed tailored PIR scheme to the usage of prior work from the PIR literature as a substitute for our PIR scheme, as a function of $\rep \leq \nclients$ for the case when $\nclients=10$, $\nobjectives=10$, and $\ncolluderss=\ncolludersq=1$. In \cref{tab:rates}, we provide a summary of the sharing and PIR rates \rev{and the communication costs} for the three schemes as a function of the system parameters. We plot the separate sharing and PIR rates for the same parameters above in \cref{fig:pir_rates,fig:sharing_rates}, which exhibit  contrasting behavior as a function of $\rep$. Hence, the computation cost and, consequently, the utility determined by $\rep$ incur a trade-off between the sharing and the PIR rates. One can further find that our proposed solution in combination with the newly designed PIR scheme outperforms the usage of existing graph-based PIR schemes in the query stage of our method. When $\rep = \nclients$, we show that the usage of a star-product-based PIR scheme in the query stage with parameters specifically optimized for MDS coded data in form of secret sharing gives the smallest communication cost. \rev{We depict the communication cost for $\nclients=100$ clients, $\nobjectives=20$, and $\ncolluderss=\ncolludersq=5$ in \cref{fig:communication_cost_100}, and find that our scheme uniformly outperforms the application of existing graph-based PIR schemes. For $\rep=\nclients=100$, using the tailored star-product-based PIR scheme is beneficial.} \rev{According to \cref{def:privacy_clients,def:objective_hiding,def:privacy_federator}, no privacy leakage is incurred by our scheme.}

\section{Conclusion}

In this work, we introduced a new notion of objective-hiding for federated one-shot learning complemented by data privacy for the clients' data. We use tools from multi-party computation, knowledge distillation and (S)PIR, and propose a new three-stage protocol that achieves information-theoretic privacy of the federator's objective and the clients' data under a limited collusion assumption. To minimize the joint communication cost in the sharing and query stages of our framework, we proposed a novel graph-based PIR scheme for flexible task assignments specifically tailored to the setting at hand and leveraging the properties of dual GRS codes. 

Going further, the problems of mitigating the effect of stragglers and dropouts among clients, considering the presence of malicious clients deliberately trying to corrupt the process, and considering different privacy notions such as differential privacy \cite{dwork2006differential} and subset privacy \cite{raviv2022perfect} remain open in this setting.

\appendix

\subsection{Proof of \cref{thm:privacy}}
We first prove the clients' data privacy from any other set $\colss$ of at most $\ncolluderss$ colluding clients. Note that each client $\clientidxb \in \colss$ receives the following set of messages from all other clients: $\allreceived[\clientidxb] \define \{\ssclient[\evalc[\clientidxb]]\}_{\client \in \incidentedge, \objective \in \incidentvertex[\clientidxb] \setminus \{\clientidxb\}, \partition \in [\npartitions]}$, i.e., for each objective $\objective \in \incidentvertex[\clientidxb]$ one share from each client that was assigned the same objective $\objective$, per partition $\partition \in [\npartitions]$. Hence, for each objective $\objective \in [\nobjectives]$ and each partition $\partition \in [\npartitions]$ and set of clients in $\colss$ receive at most $\ncolluderss$ shares. Since, by design, at most $\ncolluderss$ secret shares of any two objectives $\objective, \objective^\prime \in [\nobjectives]$ and two partitions $\partition, \partition^\prime$ are pair-wise independent of each other, it suffices to prove that for all clients $\client \in [\nclients]$ any for any pair of $\objective, \partition$, we have $\condmutinf{\{\logitpartrv\}_{\partidx \in [\kstorage-\ncolluderss]}}{\allreceivedcol[\colss]}{\{\logitrv\}_{\objective \in \incidentvertex, \client \in \colss, \labelidx\in [\nlabels]}} = 0$, where $\logitpartrv$ is the random variable corresponding to $\logitpart$ as defined in \cref{sec:scheme}. The set of critical messages is then given by $\{\ssclient[\evalc[\clientidxb]]\}_{\clientidxb \in \colss}$. Hence, we need to prove that
\begin{align*}
    &\condmutinf{\{\logitpartrv\}_{\partidx \in [\kstorage-\ncolluderss]}}{\{\ssclient[\evalc[\clientidxb]]\}_{\clientidxb \in \colss}}{\{\logitrv\}_{\objective \in \incidentvertex, \client \in \colss, \labelidx\in [\nlabels]}} \\
    &\condentropy{\{\logitpartrv\}_{\partidx \in [\kstorage-\ncolluderss]}}{\{\logitrv\}_{\objective \in \incidentvertex, \client \in \colss, \labelidx\in [\nlabels]}} \\
    &\!-\! \condentropy{\!\{\logitpartrv\}_{\partidx \in [\kstorage-\ncolluderss]}\!\!}{\!\!\{\logitrv\}_{\objective \in \incidentvertex, \client \in \colss, \labelidx\in [\nlabels]}, \! \{\ssclient[\evalc[\clientidxb]]\}_{\clientidxb \in \colss}\!} \\
    &\condentropy{\{\logitpartrv\}_{\partidx \in [\kstorage-\ncolluderss]}}{\{\logitrv\}_{\objective \in \incidentvertex, \client \in \colss, \labelidx\in [\nlabels]}} \\
    &-\condentropy{\{\logitpartrv\}_{\partidx \in [\kstorage-\ncolluderss]}}{\{\logitrv\}_{\objective \in \incidentvertex, \client \in \colss, \labelidx\in [\nlabels]}} = 0,
\end{align*}
which holds since any $\ncolluderss$ shares do not reveal anything about the privacy labels $\{\logitpartrv\}_{\partidx \in [\kstorage-\ncolluderss]}$ beyond the correlation with the colluders' information.

We now prove the privacy of the objective $\queriedobjective$ of interest in the following, in particular, the queries or shares observed and held by any set of $\ncolludersq$ clients does not leak any information about $\queriedobjective$. When clear from the context, we omit the subscripts from the set notation for readability.
\begin{align*}
    &\mutinf{\{\queryrvcol\}, \camready{\allreceivedcol[\colquery]}, \allcol}{\queriedobjective} \\
    &\overset{(a)}{=} \mutinf{\{\queryrvcol\}, \allreceivedcol[\colquery], \{\logitrv\}_{\client \in \colquery}}{\queriedobjective} \\
    &\overset{(b)}{=} \mutinf{\{\queryrvcol\}}{\queriedobjective} + \condmutinf{\allreceivedcol[\colquery], \{\logitrv\}_{\client \in \colquery}}{\queriedobjective}{\{\queryrvcol\}} \\
    &= \mutinf{\{\queryrvcol\}}{\queriedobjective} \begin{aligned}[t] &+ \condentropy{\allreceivedcol[\colquery], \{\logitrv\}_{\client \in \colquery}}{\{\queryrvcol\}}\\ &- \condentropy{\allreceivedcol[\colquery], \{\logitrv\}_{\client \in \colquery}}{\{\queryrvcol\}, \queriedobjective} \end{aligned} \\
    &\overset{(c)}{=} \mutinf{\!\{\queryrvcol\}}{\!\queriedobjective} \!+\!\entropy{\!\allreceivedcol[\colquery], \!\! \{\logitrv\}_{\client \in \colquery}} \!-\! \entropy{\!\allreceivedcol[\colquery], \!\!\{\logitrv\}_{\client \in \colquery}\!}
\end{align*}
where $(a)$ is by the definition of $\allcol$, $(b)$ follows from the chain rule of mutual information, $(c)$ from the definition of conditional mutual information and by independence of $\{\queryrvcol\}, \queriedobjective$ of $\allreceivedcol[\colquery],\{\logitrv\}_{\client \in \colquery}$.

Assuming w.l.o.g. that $\queriedobjective = 1$, then
\begin{align*}
    &\mutinf{\{\queryrvcol\}, \allcol}{\queriedobjective} = \sumo \condmutinf{\{\queryrvcol\}}{\queriedobjective}{\{\queryrvcoltmp\}_{\objectivetmp < \objective}} \\
    &\overset{(d)}{=} \sumo \mutinf{\{\queryrvcol\}}{\queriedobjective} \overset{(e)}{\leq} \sumo \mutinf{\{\queryrvcol\}}{\querymessage} = 0,
\end{align*}
where $(d)$ is because any set of at most $\ncolludersq$ shares $\queryrvcol$ encoding the query for objective $\objective$ is independent of another set of secret shares $\queryrvcol$ encoding the query for objective $\objectivetmp$ by the properties of secret sharing. $(e)$ holds by the data processing inequality, and the last equality holds since \eqref{eq:query_poly} is a secret sharing according to McEliece-Sarwate \cite{mceliece1981sharing} where any set of at most $\ncolludersq$ shares are statistically independent of the private message $\querymessage$. This concludes the proof.

\subsection{Proof of \cref{thm:correctness}}
\begin{proof}

\allowdisplaybreaks
The federator receives the answers $\answer = \{\answerpart[1], \cdots, \answerpart[\npartitions]\}$, where $\answerpart$ are evaluations of the following answer polynomial
\begin{align*}
    &\ssanswerpart = \sumi \dualco \sumssclient[x] \ssquery[x] \\
    &= \begin{aligned}[t] \sumi \dualco &\left(\sum_{\partidx=1}^{\kstorage-\ncolluderss} \!\!\! x^{\partidx\!-\!1} \! \sum_{\client \in \incidentedge} \logitpart[\partidx] + \sum_{\tau=1}^{\ncolluderss} x^{\kstorage-\ncolluderss+\tau-1} \rlogitssum[\tau] \right) \\
    &\left(\querymessage[\partidx] + \sum_{\tau=1}^{\ncolludersq} x^{\kstorage-\ncolluderss+\tau-1} \rquery[\tau]\right) \end{aligned} \\
    &= \sumi \dualco \left(\sum_{\partidx=1}^{\kstorage-\ncolluderss} \!\!\! x^{\partidx\!-\!1} \! \sum_{\client \in \incidentedge} \logitpart[\partidx]\right) \querymessage[\partidx] \\
    &+ \sumi \dualco \left(\sum_{\partidx=1}^{\kstorage-\ncolluderss} \!\!\! x^{\partidx\!-\!1} \! \sum_{\client \in \incidentedge} \logitpart[\partidx]\right) \!\! \left(\sum_{\tau=1}^{\ncolludersq} x^{\kstorage-\ncolluderss+\tau-1} \rquery[\tau]\right) \\
    &+ \sumi \dualco \left(\sum_{\tau=1}^{\ncolluderss} x^{\kstorage-\ncolluderss+\tau-1} \rlogitssum[\tau]\right) \querymessage[\partidx] \\
    &+ \sumi \dualco \left(\sum_{\tau^\prime=1}^{\ncolluderss} x^{\kstorage-\ncolluderss+\tau^\prime-1} \rlogitssum[\tau]\right) \!\! \left(\sum_{\tau=1}^{\ncolludersq} x^{\kstorage-\ncolluderss+\tau-1} \rquery[\tau]\right) \\
    &= \dualcoqueried \sum_{\partidx=1}^{\kstorage-\ncolluderss} x^{\partidx-1} \sum_{\client\in \incidentedge[\queriedobjective]} \logitpartqueried[\partidx] + \dualcoqueried \sum_{\tau=1}^{\ncolluderss} x^{\kstorage-\ncolluderss+\tau-1} \rlogitssumqueried[\tau] \\
    &+ \sum_{\partidx=1}^{\kstorage-\ncolluderss} \sum_{\tau=1}^{\ncolludersq} x^{\kstorage-\ncolluderss+\partidx+\tau-2} \sumi \dualco \left(\rquery[\tau] \sum_{\client \in \incidentedge} \logitpart[\partidx]\right) \\
    &+ \sum_{\tau^\prime=1}^{\ncolluderss} \sum_{\tau=1}^{\ncolludersq} x^{2\kstorage-2\ncolluderss+\tau^\prime+\tau-2} \sumi \dualco \rlogitssum[\tau^\prime] \rquery[\tau], %
    \end{align*}

For $\evalidx \in [\kstorage-\ncolluderss]$, summing over all clients, we have
\begin{align*}
    &\sumc \evalc^{-\evalidx} \answerpart = \sumc \sumi \evalc^{-\evalidx} \dualco \sumssclient[\evalc] \ssquery[\evalc] \\
    &= \sumo \sumio \evalc^{-\evalidx} \dualco \sumssclient[\evalc] \ssquery[\evalc] \\
    &= \sum_{\partidx=1}^{\kstorage-\ncolluderss} \sumioqueried\left(\dualcoqueried \evalc^{\partidx-\evalidx-1} \sum_{\client\in \incidentedge[\queriedobjective]} \logitpartqueried[\partidx]\right) \\
    &+ \sum_{\tau=1}^{\ncolluderss} \rlogitssumqueried[\tau] \sumioqueried \dualcoqueried \evalc^{\kstorage-\ncolluderss+\tau-\evalidx-1} \\
    &+ \sumo \! \sum_{\partidx=1}^{\kstorage-\ncolluderss} \! \sum_{\tau=1}^{\ncolludersq} \left(\rquery[\tau] \!\!\! \sum_{\client \in \incidentedge} \!\!\! \logitpart[\partidx]\right) \!\! \sumio \!\!\! \dualco \evalc^{\kstorage-\ncolluderss+\partidx+\tau-\evalidx-2} \\
    &+ \sumo \sum_{\tau^\prime=1}^{\ncolluderss} \sum_{\tau=1}^{\ncolludersq} \rlogitssum[\tau^\prime] \rquery[\tau] \sumio \dualco \evalc^{2\kstorage-2\ncolluderss+\tau^\prime+\tau-\evalidx-2} \\
    &= \sum_{\partidx=1}^{\kstorage-\ncolluderss} \left(\sum_{\client\in \incidentedge[\queriedobjective]} \logitpartqueried[\partidx]\right) \sumioqueried \dualcoqueried \evalc^{\partidx-\evalidx-1} \\
\end{align*}
where the latter step holds because $\sumio \dualco \evalc^\zeta = 0$ for all $0 \leq \zeta \leq \rep-2$. It is ensured that $\kstorage-\ncolluderss \leq \frac{\rep-\ncolluderss-\ncolludersq+1}{2}$. Consequently, we have 1) that $0 \leq \tau-1 \leq \kstorage-\ncolluderss+\tau-\evalidx-1 \leq \kstorage-\ncolluderss+\tau-2 \leq \kstorage-\ncolluderss-3$, 2) that $2 \leq \partidx+\tau \leq \kstorage-\ncolluderss+\partidx+\tau-\evalidx-2 \leq \kstorage-\ncolluderss+\partidx+\tau-3 \leq 2(\kstorage-\ncolluderss)+\ncolludersq-3 \leq \rep - \ncolluderss -2 \leq \rep-2$, and 3) that $\kstorage-\ncolluderss \leq \kstorage-\ncolluderss+\tau^\prime+\tau-2 \leq 2\kstorage-2\ncolluderss+\tau^\prime+\tau-\evalidx-2 \leq 2\kstorage-2\ncolluderss+\tau^\prime+\tau-3 \leq 2\kstorage-2\ncolluderss+\ncolluderss+\ncolludersq-3 \leq \rho-2$. Hence, in all cases, $\zeta$ is between $0$ and $\rep-2$, which is why all terms except the first vanish.

Similarly, all terms corresponding to $\evalc^{\partidx-\evalidx-1}$ vanish when $\evalidx < \partidx$, and hence we obtain from $\evalc \in [\kstorage-\ncolluderss]$ the following linear system of equations. For notational convenience, we let $\logitpartqueriedsum \define \sum_{\client\in \incidentedge[\queriedobjective]} \logitpartqueried[\partidx]$. Further, let $\evalsum \define \sumc \evalc^{-\evalidx} \answerpart$. Without loss of generality, let clients $\client \in [\rep]$ store the logits for objective $\queriedobjective$. We have for
\begin{align*}
\mathbf{P} &= \sumioqueried \begin{pmatrix}
    \dualcoqueried \evalc^{-1} & 0 & \cdots & 0 & 0 \\
    \dualcoqueried \evalc^{-2} & \dualcoqueried \evalc^{-1} & 0 & \cdots & 0 \\
    \vdots & \ddots & \ddots & \ddots & 0 \\
    \dualcoqueried \evalc^{-\kstorage+\ncolluderss} & \cdots & \cdots & \cdots & \dualcoqueried \evalc^{-1}
\end{pmatrix} & \\
\end{align*}
that~\begin{align*}
    \begin{pmatrix}
        \evalsum[1] \\
        \evalsum[2] \\
        \cdots \\
        \evalsum[\kstorage-\ncolluderss]
    \end{pmatrix} = 
    \mathbf{P}
    \begin{pmatrix}
        \logitpartqueriedsum[1] \\
        \logitpartqueriedsum[2] \\
        \cdots \\
        \logitpartqueriedsum[\kstorage-\ncolluderss]
    \end{pmatrix}
\end{align*}
from which the desired values $\logitpartqueriedsum[1], \logitpartqueriedsum[2], \cdots, \logitpartqueriedsum[\kstorage-\ncolluderss]$ can be obtained. Using the triangular shape of the matrix and applying \cref{lemma:nonzero_proof} proves that $\mathbf{P}$ is invertible when $q > \rep + \kstorage-\ncolluderss -1$, and $\evalc, \client \in [\nclients],$ are chosen as $\evalc = \eval^\client, \client \in [\nclients]$ for a generator $\eval$ of the field $\F_q$. We state and prove the lemma in the following. This also concludes the proof of the theorem.
\begin{lemma} \label{lemma:nonzero_proof}
    For $q > \rep + \kstorage-\ncolluderss -1$, and $\evalc = \eval^\client, \client \in [\nclients]$, for any generator element $\eval$ of the field $\F_q$, we have $\sumioqueried \dualcoqueried \evalc^{-\evalidx} \neq 0$ for all $\evalidx \in [\kstorage-\ncolluderss]$.
\end{lemma}
\begin{proof}
    Assuming without loss of generality that $\incidentedge[\queriedobjective] = [\rep]$, we write $\sumioqueried \dualcoqueried \evalc^{-\evalidx}$ as the inner product of two codewords:
    \begin{align*}
        \sumioqueried \dualcoqueried \evalc^{-\evalidx} = (\dualcoqueried[1], \cdots, \dualcoqueried[\rep]) (\evalc[1]^{-\evalidx}, \cdots, \evalc[\rep]^{-\evalidx})^T
    \end{align*}
    Let $(\dualcoqueried[1], \cdots, \dualcoqueried[\rep]) \in \mathcal{C}$, where $\mathcal{C}$ is a generalized Reed-Solomon code with dimension $1$ and length $\rep$. If and only if $(\evalc[1]^{-\evalidx}, \cdots, \evalc[\rep]^{-\evalidx}) \in \mathcal{C}^\perp$, it holds that $(\dualcoqueried[1], \cdots, \dualcoqueried[\rep])(\evalc[1]^{-\evalidx}, \cdots, \evalc[\rep]^{-\evalidx})^T = 0$. We need to show that $(\evalc[1]^{-\evalidx}, \cdots, \evalc[\rep]^{-\evalidx}) \notin \mathcal{C}^\perp$. By the definition of generalized Reed-Solomon codes and their dual, the generator matrix $\mathbf{G}_{\mathcal{C}^\perp}$ of $\mathcal{C}^\perp$ is given by
    \begin{align}
        \mathbf{G}_{\mathcal{C}^\perp} = \begin{pmatrix}
            1 & 1 & \cdots & 1 \\
            \evalc[1] & \evalc[2] & \cdots & \evalc[\rep] \\
            \vdots & \vdots & \ddots & \vdots \\
            \evalc[1]^{\rep-2} & \evalc[2]^{\rep-2} & \cdots & \evalc[\rep]^{\rep-2}
        \end{pmatrix} \label{eq:Gdual}
    \end{align}
    Hence, $(\dualcoqueried[1], \cdots, \dualcoqueried[\rep])(\evalc[1]^{-\evalidx}, \cdots, \evalc[\rep]^{-\evalidx})^T = 0$ for all $2-\rep \leq \evalidx \leq 0$ (which we used above).
    Since we have $\evalc^{q-1}=1$, we can write 
    \begin{align*}
        \sumioqueried \dualcoqueried \evalc^{-\evalidx} &= (\dualcoqueried[1], \cdots, \dualcoqueried[\rep]) (\evalc[1]^{-\evalidx}, \cdots, \evalc[\rep]^{-\evalidx})^T \\
        &= (\dualcoqueried[1], \cdots, \dualcoqueried[\rep]) (\evalc[1]^{q-\evalidx-1}, \cdots, \evalc[\rep]^{q-\evalidx-1})^T
    \end{align*}
    If $q-\evalidx-1 \leq \rep-2$, then from $\eqref{eq:Gdual}$, it can be seen that the above inner product is $0$. If $q-\evalidx-1 > \rep-2$, we must show that the matrix
    \begin{align*}
        \begin{pmatrix}
            1 & 1 & \cdots & 1 \\
            \evalc[1] & \evalc[2] & \cdots & \evalc[\rep] \\
            \vdots & \vdots & \ddots & \vdots \\
            \evalc[1]^{\rep-2} & \evalc[2]^{\rep-2} & \cdots & \evalc[\rep]^{\rep-2} \\
            \evalc[1]^{q-\evalidx-1} & \evalc[2]^{q-\evalidx-1} & \cdots & \evalc[\rep]^{q-\evalidx-1}
        \end{pmatrix}
    \end{align*}
     is full rank for every $\evalidx \in [1, \kstorage-\ncolluderss]$. Note that only the first $\rep-1$ rows correspond to the structure of a transposed Vandermonde matrix, hence the latter row is not necessarily linearly independent from the first $\rep-1$ rows. By choosing each evaluation point $\evalc$ to be the $\client$-th power of a primitive element $\eval$, we can rewrite the matrix as
     \begin{align*}
        \begin{pmatrix}
            1 & 1 & \cdots & 1 \\
            \eval & \eval^2 & \cdots & \eval^{\rep} \\
            \vdots & \vdots & \ddots & \vdots \\
            \eval^{\rep-2} & \eval^{2(\rep-2)} & \cdots & \eval^{\rep(\rep-2)} \\
            \eval^{q-\evalidx-1} & \eval^{2(q-\evalidx-1)} & \cdots & \eval^{\rep(q-\evalidx-1)}
        \end{pmatrix},
    \end{align*}
    which exhibits the structure of a Vandermonde matrix. To satisfy that all powers $0, \cdots, q-\evalidx-1$ of $\eval$ are distinct, we require that $\texttt{ord}(\eval) > q-2 \geq q-\evalidx-1$. Hence, the order of $\alpha$ must be $\texttt{ord}(\eval) \geq q-1$, which is satisfied if $\eval$ is a generator of the field $\F_q$. 
    In this case, the above matrix is Vandermonde, thus full rank with all rows and columns being linearly independent. Hence, $(\eval^{q-\evalidx-1}, \eval^{2(q-\evalidx-1)}, \cdots, \eval^{\rep(q-\evalidx-1)})$ is linearly independent of the rows in $\mathbf{G}_{\mathcal{C}^\perp}$, and hence $(\evalc[1]^{q-\evalidx-1}, \evalc[2]^{q-\evalidx-1}, \cdots, \evalc[\rep]^{q-\evalidx-1}) \notin \mathcal{C}^\perp$. For $\evalidx \in [1, \kstorage-\ncolluderss]$, we have $q-\evalidx-1 \geq q-\kstorage+\ncolluderss-1$. Hence, the statement holds for $q-\kstorage+\ncolluderss-1 > \rep -2$, and thus for $q > \rep + \kstorage-\ncolluderss -1$. %
\end{proof}
This concludes the proof of \cref{thm:correctness}.
\end{proof}

\subsection{Proof of \cref{lemma:optimal_comm}}
\begin{proof}
Considering the rate of the secret sharing stage and the PIR scheme following the construction in \cite{freij2017private}, the communication cost in $\F_{\nclasses+1}$ is given by
\begin{equation*}
    \frac{\nobjectives \nlabels \nclasses \cdot \nclients (\nclients-1)}{\kstorage-\ncolluderss} + \frac{\nlabels \nclasses}{\kstorage-\ncolluderss} \frac{\kstorage \nclients}{\nclients - \kstorage - \ncolludersq +1},
\end{equation*}
which is a convex optimization over the convex set $\ncolludersq+1 \leq \kstorage \leq \nclients - \ncolludersq$ since being a sum of convex functions over that set. For $c_1 = (\nclients - \ncolludersq + 1)$ and $c_2 = \nobjectives \nclients (\nclients-1)$, we have
\begin{align*}
    &\frac{\partial}{\partial \kstorage} \frac{\nobjectives \nlabels \nclasses \cdot \nclients (\nclients-1)}{\kstorage-\ncolluderss} + \frac{\nlabels \nclasses}{\kstorage-\ncolluderss} \frac{\kstorage \nclients}{\nclients - \kstorage - \ncolludersq +1} \\
    &= \nlabels \nclasses \left( -\frac{\nobjectives \nclients (\nclients-1)}{(\kstorage - \ncolluderss)^2} + \frac{\nclients (\kstorage^2 - \ncolluderss(\nclients-\ncolludersq+1))}{(\kstorage-\ncolluderss)^2 (\nclients-\kstorage-\ncolludersq+1)^2} \right) \\
    &= \nlabels \nclasses \left( -\frac{c_2 \left(c_1^2 + \kstorage^2 - 2 c_1 \kstorage \right) - \nclients (\kstorage^2 - c_1 \ncolluderss)}{(\kstorage-\ncolluderss)^2 (\nclients-\kstorage-\ncolludersq+1)^2} \right) \\
    &= \nlabels \nclasses \nclients \left( -\frac{ \kstorage^2 (c_2 - \nclients) - 2\kstorage c_1 c_2 + c_1^2 c_2 + \nclients c_1 \ncolluderss}{(\kstorage-\ncolluderss)^2 (\nclients-\kstorage-\ncolludersq+1)^2} \right).
\end{align*}
Setting the derivative to zero, we obtain
\begin{align*}
    \kstorage^\prime &\!=\! \frac{1}{2(c_2 - \nclients)} \!\! \left(\! 2 c_1 c_2 \! -\! \sqrt{4c_1^2 c_2^2 - 4 (c_2 - \nclients) (c_1^2 c_2 \! + \! \nclients c_1 \ncolluderss)}\right)\!,
\end{align*}
and hence the statement above.
\end{proof}

\balance
\bibliographystyle{IEEEtran}
\bibliography{sample,refs}

\vfill

\end{document}